\documentclass[12pt]{article}

\usepackage{makeidx}  
\usepackage{amsmath,amssymb,amsfonts,latexsym,epsfig}
\usepackage{psfrag, xspace}
\usepackage{amsthm}

\usepackage[colorlinks,linkcolor=black,citecolor=black]{hyperref}
\usepackage{colortbl}
\usepackage{caption}
\usepackage{geometry}
\usepackage{blkarray,arydshln}
\usepackage{multirow}
\usepackage{mathrsfs}
\usepackage{graphicx}  
\usepackage{float}
\usepackage{cases}
\usepackage{url}

\usepackage{bm}
\usepackage{indentfirst}
\usepackage{enumerate}

\usepackage{algorithm}               
\usepackage{algorithmic}             
\usepackage{multirow}                

\newtheorem{theorem}{Theorem}

\newtheorem{exm}{Example}
\newtheorem{defin}[theorem]{Definition}

\newtheorem{lem}[theorem]{Lemma}
\newtheorem{thm}[theorem]{Theorem}

\theoremstyle{definition}{\bf}
\newtheorem*{rem}{Remarks}

\newtheorem*{key}{Key words}

\def\C {{\mathbb{C}}}
\def\x{\mathbf{x}}
\def\X{\mathbf{X}}
\def\p{\mathbf{p}}
\def\q{\mathbf{q}}
\def\T{\mathrm{T}}
\def\llambda{\boldsymbol{\lambda}}
\def\alphahei{\boldsymbol{\alpha}}
\def\0{{\mathbf{ 0}}}
\def\N{{\mathbb{N}}}

\def\rank{{\mathrm{rank}}}

\def\V{{\mathbb{V}}}
\def\F{{\mathbf{F}}}
\def\G{{\mathbf{G}}}
\def\H{{\mathbf{H}}}
\def\J{{\boldsymbol{\mathrm{J}}}}

\title{A New Deflation Method For Verifying the Isolated Singular Zeros of Polynomial Systems}
\author{Jin-San Cheng$^{b,c}$, \hspace{5mm} Xiaojie Dou$^{a}$, \hspace{5mm}  Junyi Wen$^{b,c}$
\\
$^{a}$College of Science, Civil Aviation University of China, Tianjin\\
 $^{b}$KLMM, Academy of Mathematics and Systems Science\\
Chinese Academy of Sciences, Beijing\\
$^{c}$University of Chinese Academy of Sciences, Beijing\\
jcheng@amss.ac.cn\hspace{2mm} xjdou@amss.ac.cn\hspace{2mm}  wenjunyi15@mails.ucas.ac.cn
}
\date{\today}

\begin{document}

\maketitle
        \begin{abstract}
In this paper, we develop a new deflation technique for refining or verifying the isolated singular zeros of polynomial systems. Starting from a polynomial system with an isolated singular zero, by computing the derivatives of the input polynomials directly or the linear combinations of the related polynomials, we construct a new system, which can be used to refine or verify the isolated singular zero of the input system. In order to preserve the accuracy in numerical computation as much as possible, new variables are introduced to represent the coefficients of the linear combinations of the related polynomials. To our knowledge, it is the first time that considering the deflation problem of polynomial systems from the perspective of the linear combination. Some acceleration strategies are proposed to reduce the scale of the final system. We also give some further analysis of the tolerances we use, which can help us have a better understanding of our method.
%
%
The experiments show that our method is effective and efficient. Especially, it works well for zeros with high multiplicities of large systems. It also works for isolated singular zeros of non-polynomial systems.

\begin{key}
Polynomial system, deflation method, isolated singular zero, interval verification
\end{key}
\end{abstract}

\section{Introduction}

  Solving polynomial systems with singular zeros is always a challenge in algebraic and geometric computation. For an isolated simple zero of a polynomial system, the classical Newton's method is widely used and quadratic convergent. However, for singular zeros of a polynomial system, Newton's method is not fit for the original system directly because it converges slowly or even doesn't converge in a bad situation. 
  %
 What's more, it is an ill-posed problem to compute an isolated singular zero of a polynomial system or a nonlinear system, since a small perturbation of coefficients may transform an isolated singular zero into a cluster of simple zeros.

 Therefore, finding methods to keep the quadratic convergence of Newton's method for singular zeros is a way to handle this problem.
 Given a polynomial system with an isolated singular zero, we can construct a new system owing the same singular zero as an isolated simple one. Based on this idea, in recent years, there are many symbolic or symbolic-numerical methods coming up to deal with this problem. The basic idea is the deflation techniques 
 \cite{dayton,giusti1,giusti2,lecerf,hao,hauensten,hau,zeng}, which usually have two basic strategies: adding new equations only or both new equations and new variables to the original system.

 Deflation for an isolated singular solution originated from the ideas of Ojika \cite{Ojika,Ojika1,Ojika2}. T. Ojika et al. present a deflation algorithm for determining the multiple zeros for a system of nonlinear equations. Through triangulating the Jacobian matrix of the original system at an approximate zero, new equations, which comes from the minors of the Jacobian matrix, are introduced to the original system to reduce the multiplicity until they get a system which is regular at the singular zero.

 In \cite{giusti3}, Giusti and Yakoubsohn propose a construction, which is based on two operations: deflating and kerneling, to determine a regular system without adding new variables. In the deflating, all the partial derivatives of the polynomials, which are zero at the multiple zero, are introduced to replace the corresponding polynomials. The kerneling operation consists of adding the polynomials given by the nonzero numerators of the coefficients of the Schur complement of the Jacobian matrix of the original system to the original system.

 In \cite{hauensten}, Hauenstein and Wampler define a strong deflation by only adding new equations coming from the one order differential of the Jacobian matrix of the original system to the original system. Different from \cite{giusti3}, at each iteration step, both the number and the degree of the added equations are reduced.

 In \cite{mourrain1}, Mourrain et al. give a method which uses a single linear differential form defined from the Jacobian matrix of the input system, and defines the deflated system by applying this differential form to the original system.


 These above methods do introduce new equations and finally get a new system owing the isolated singular zero of the original system as a simple zero. In order to get the new polynomials, one needs to compute the determinant of some polynomial matrices. Thus the degree of the polynomials in the new system may be very high.

 In the following, denote $n$ as the number of both the variables and the equations in the original system, and $\mu$ as the multiplicity of the isolated singular zero of the original system.

  In \cite{Yama}, Yamamoto introduces new equations and new variables to the original system simultaneously.  New variables are used to bring some perturbations of the original system and the Jacobian matrix of the original system, which produce new equations.

 In \cite{lvz,lvz1,lvz0}, Leykin et al. present an effective modification of Newton's method to restore quadratic convergence for isolated singular solutions of polynomial systems. Different from \cite{Yama}, new variables are only introduced to the Jacobian matrix of the original system, which produce new equations.
 Meanwhile, they also prove that the number of deflation stages is bounded by $\mu$.

 In \cite{dayton1}, Dayton and Zeng modify the method in \cite{lvz} and further prove that the number of deflation steps is bounded by the depth of the dual space. For the special case of breathe one, they also propose a modified deflation method, which is based on duality analysis, to reduce the final size $2^{\mu-1}n\times2^{\mu-1}n$ of deflated system in \cite{lvz} to $\mu n\times \mu n$.

 In \cite{rump1}, by introducing a smoothing parameter to the original system and $n-1$ new variables to the Jacobian matrix of the original system, which produces new equations, Rump and Graillat consider the case of the double zero of the original system. In \cite{zhi2}, based on the parameterized multiplicity structure, Li and Zhi generalize the algorithm in \cite{rump1} to deflate the breath-one isolated singular zero of the original system. Their final deflated regular system is of size $\mu n\times \mu n$.


 In \cite{mourrain}, based on the given multiplicity structure of the original system, which depends on the accuracy of the given approximate multiple zero, Mantzaflaris and Mourrain give a method to find a (small) perturbed system of the original system 
 and then first compute a deflated system in one deflation step. The size of the final deflated system is equal to $\mu n\times \mu n$.


 In \cite{zhi3}, by lifting the independent perturbations in the first-order differential system appearing in \cite{Yama} back to the original system, Li and Zhi modify the method in \cite{Yama} and also prove that the modified deflation technique terminates after a finite number of steps bounded by the depth of the dual space. The size of the final modified regularized system is bounded by $2^{\mu-1}n\times2^{\mu-1}n$.


 In \cite{mourrain1}, by introducing some variables to represent the coefficients of the dual basis, Mourrain et al. give a method to deflate the original system and determine the multiplicity structure simultaneously. They also show that the number of variables and equations in this method is bounded by $n+n\mu(\mu-1)/2$ and $n\mu+n(n-1)(\mu-1)(\mu-2)/4$. However, one point worth noting is that this method needs to know the monomial basis of the original system first.

 These methods introduce new variables and new equations to the original system simultaneously. By repeatedly using these deflation constructions, they will get an augmented system finally, which has an isolated simple zero, whose partial projection corresponds to the isolated singular zero of the original system.

 {\bf Main contribution.}  In this paper,  given a polynomial system $\F\subset\C[\x]$ with an isolated singular zero $\p$, by computing the derivatives of the input polynomials directly or the linear combinations of the related polynomials, we propose a new deflation method to construct a final deflated system $\widetilde{\F}'(\x,\alphahei)$, which has an isolated simple zero $(\p,\hat{\alphahei})$, whose projection corresponds to the isolated singular zero $\p$ of the input system. New variables $\alphahei$ are introduced to represent the coefficients of the linear combinations of the related polynomials to ensure the accuracy of the numerical implementation. Moreover, we also prove that the size of our deflation system depends on the depth or the multiplicity of $\p$.

 Compared to the previous methods, our method has the following differences:
 \begin{enumerate}
   \item For the input system $\F$, we can, if needed, compute the derivatives of every $f_i$ to get the needed polynomials, which are regular at $\p$ at the beginning. Then, we put all these polynomials together to construct a system $\F_0$, such that the rank $r$ of its Jacobian matrix at $\p$ is maximal. In some cases, we have $r = n$, which means that we need not introduce new variables.
   \item We compute the derivatives of the linear combinations of the related polynomials to get some polynomials which are regular at $\p$. Here we introduce new variables to represent the coefficients of the linear combination.
   \item Considering that we  know only the approximate zero $\tilde{\p}$ in actual computations, we use a tolerance $\theta$ to judge if a polynomial is $\theta$-regular or $\theta$-singular at $\tilde{\p}$ and another tolerance $\varepsilon$ to judge the numerical rank of the Jacobian matrix. As long as the tolerance $\theta$ is chosen properly, we will get the same judgement in numerical case as in the exact case. Thus, our deflation system usually has the same isolated zero as the input system. Inspired by the work \cite{lvz} of Leykin et al., we also give some further analysis on the tolerances $\theta$ and $\varepsilon$, which tells us that our final system is a perturbed system with a bounded perturbation in the worst case. To make our final system as accurate as possible, we also analyse the case that the tolerance $\theta$ is not introduced.
 \end{enumerate}

 Thanks to the above acceleration strategies, the size of the final system in our actual computations is much less than that we give in theory.
 Furthermore, we implement our method in Matlab.
%
 The experiments show that our method is effective and efficient, especially for large systems with singular zeros of high multiplicities. Besides, for the non-polynomial systems, our method is also applicable.

 The paper is organized as below. We introduce some notations and preliminaries in the next section. In Section 3, we give a new deflation idea to construct a deflated square system from the input system with an isolated singular zero. An effective version of our method is given in Section 4. Some numerical experiment results are given to demonstrate the performance of our algorithm in Section 5 and at last, we draw a conclusion in Section 6.

\section{Notations and Preliminaries}

Let $\mathbb{C}$ be the complex field and $\mathbb{C}[\mathbf{x}]=\mathbb{C}[x_1,\ldots,x_n]$ be the polynomial ring. Denote $\F=\{ f_1,f_2,\ldots,f_n\}\subset \C[\x]$ as a polynomial system and $\deg(f_i)$ as the degree of the polynomial $f_i$. Similarly, $\deg(\F)=\max\limits_{f_i\in\F}{\deg(f_i)}$. Let $\mathbf{p}=(p_1,\ldots,p_n)\in\mathbb{C}^n$. $\F(\mathbf{p})=\0$ denotes that $\mathbf{p}$ is a zero of $\F(\x)=\mathbf{0}$.

Let $\V(\F)\subset \C^n$ denote the variety defined by $\F$ and  $\dim\V(\F)$ denote the dimension of $\V(\F)$.

Let $\mathbf{d}^{\boldsymbol{\gamma}}_\mathbf{x}:\mathbb{C}[\mathbf{x}]\rightarrow \mathbb{C}[\x]$ denote the differential functional defined by
$$\mathbf{d}^{\boldsymbol{\gamma}}_\mathbf{x}(f)=\frac{1}{\gamma_1!\cdots\gamma_n!}\cdot\frac{\partial^{|{\boldsymbol{\gamma}}|}f}{\partial x^{\gamma_1}_1\cdots\partial x^{\gamma_n}_n}, \qquad \forall f\in\mathbb{C}[\mathbf{x}],$$
where ${\boldsymbol{\gamma}}=(\gamma_1,\ldots,\gamma_n)\in \mathbb{N}^n$ with $\mathbb{N}=\{0,1,2,\ldots\}$ and $|\boldsymbol{\gamma}|=\sum\limits_{i=1}^n \gamma_i$.

Denote $\rank(A)$ as the rank of a matrix $A$. Denote $\J(\F)$ as the Jacobian matrix of $\F$. That is,

$$
\J(\F)=
\left(
\begin{array}{ccc}
  \frac{\partial f_1}{\partial x_1} & \ldots & \frac{\partial f_1}{\partial x_n} \\
  \vdots & \ddots & \vdots \\
  \frac{\partial f_n}{\partial x_1} & \ldots & \frac{\partial f_n}{\partial x_n}\\
\end{array}
\right).
$$

For a polynomial $f\in\mathbb{C}[\mathbf{x}]$, let $\J(f)$ denote $(\frac{\partial f}{\partial x_1},\frac{\partial f}{\partial x_2},\ldots,\frac{\partial f}{\partial x_n})$ and $\J_i(f)=\frac{\partial f}{\partial x_i}$. Let $\boldsymbol{\mathrm{J}}(\F)(\mathbf{p})$ denote the value of a function matrix $\J(\F)$ at a point $\mathbf{p}$, similarly for $\J(f)(\mathbf{p})$.

\begin{defin}
 An {\bf isolated zero} of $\F(\mathbf{x})=\0$ is a point $\mathbf{p}\in \mathbb{C}^n$ which satisfies:
$$\exists\ \varepsilon>0: \{\mathbf{y}\in\mathbb{C}^n:\|\mathbf{y}-\mathbf{p}\|<\varepsilon\}\cap \F^{-1}(\mathbf{0})=\{\mathbf{p}\},$$
where $\F^{-1}(\mathbf{0})\triangleq\{\p\in\C^n : \F(\p)=\0\}$.
\end{defin}

\begin{defin}
 We call an isolated zero $\mathbf{p}\in\C^n$ of $\F(\mathbf{x})=\0$ an {\bf isolated singular zero}  if and only if
 $$\rank(\J(\F)(\mathbf{p}))<n.$$
 Otherwise, $\p$ is an {\bf isolated regular( simple) zero} of $\F(\mathbf{x})=\0$.
\end{defin}

The {\bf Taylor series expansion} (Taylor expansion for short) of $f\in \C[\x]$ at $\p=(p_1,\ldots,p_n)\in\C^n$ is
\begin{equation}\label{eqn-taylor}
f(\x)=f(\p)+\sum_{j=1}^{n}\frac{\partial f(\p)}{\partial x_j}(x_j-p_j)+\sum_{1\le i,j\le n}\frac{\partial^2 f(\p)}{\partial x_i \partial x_j}(x_i-p_i)(x_j-p_j)+\ldots.
\end{equation}

\begin{defin} \label{def-reg}
 Let ${\p}\in \C^n$ and $f(\p)=0$. We say $f\in\C[\x]$ is {\bf singular} at $\p$ if
 $$\frac{\partial f(\p)}{\partial x_j}=0, \forall 1\le j\le n.$$
 Otherwise, we say $f$ is {\bf regular} at $\p$.
 \end{defin}

\begin{defin} \label{def-appreg}
Let $f\in\C[\x]$, $\tilde{\p}\in \C^n$ and a tolerance $\theta>0$, s.t. $|f(\tilde{\p})|<\theta$. We say $f$ is {\bf $\theta$-singular} at $\tilde{\p}$ if
 $$\left|\frac{\partial f(\tilde{\p})}{\partial x_j}\right|<\theta, \forall 1\le j\le n.$$
 Otherwise, we say $f$ is {\bf $\theta$-regular} at $\tilde{\p}$.
 \end{defin}

 \begin{lem}\label{lem-partial}
 Let $f\in\mathbb{C}[\mathbf{x}]\setminus \C$, s.t. $f(\mathbf{p})=0$. Then there exists at least a $\boldsymbol{\gamma}\in \mathbb{N}^n$, s.t. $\mathbf{d}^{\boldsymbol{\gamma}}_\mathbf{x}(f)$ is regular at $\p$.
\end{lem}
\begin{proof}
 Without loss of generality, we assume $\p=\0$.  Then $f$ can be rewritten as a sum of homogeneous polynomials as
 $$f=\sum\limits_{d=1}^{deg(f)}f_d.$$
 Since $f\not\equiv 0$, there exists at least a $\boldsymbol{\gamma}'\in \N^n$ such that $\mathbf{d}^{\boldsymbol{\gamma}'}_\mathbf{x}(f)(\p)\neq 0$. Thus there exists at least a $\boldsymbol{\gamma} \in \N^n$ such that $\mathbf{d}^{\boldsymbol{\gamma}}_\mathbf{x}(f)$ is regular at $\p$.
\end{proof}

 Now, we give an example to explain this and illustrate Definitions \ref{def-reg}, \ref{def-appreg} and Lemma \ref{lem-partial}.
\begin{exm}
Let $f=x_1+3\,x_3+4\,x_4-x_1^2+x_3^2-x_4^2-x_2^3$. For the exact point $\p=(0,0,0,0)$, we have the Taylor expansion of $f$ at $\p$ is:
$$f(\x)=x_1+3\,x_3+4\,x_4-x_1^2+x_3^2-x_4^2-x_2^3.$$
 Since $$|f(\p)|=0,\ |\J_2(f)(\p)|=0,\ |\J_i(f)(\p)|\neq0,\ i=1,3,4,$$ we know that $f$ is regular at $\p$. So is $\mathbf{d}^{(0,2,0,0)}_{\x}(f)=-3\,x_2$.

Similarly, for the approximate point $\tilde{\p}=(0.001,-0.001,0.002,-0.001)$ and a tolerance $\theta=0.01$, we have:
%
\begin{equation*}
\begin{split}
   f(\x)= & 0.003002001+0.998(x_1-0.001)-3\cdot10^{-6}(x_2+0.001)+3.004(x_3-0.002) \\
     & +4.002(x_4+0.001)-(x_1-0.001)^2+3\cdot10^{-3}(x_2+0.001)^2+(x_3-0.002)^2 \\
     & -(x_4+0.001)^2-(x_2+0.001)^3.
\end{split}
\end{equation*}

\%begin{equation}
%
 %
 Since $|f(\tilde{\p})|= 0.003002001<\theta$, $|\frac{\partial f}{\partial x_2}(\tilde{\p})|=3\cdot10^{-6}<\theta$, $|\frac{\partial f}{\partial x_1}(\tilde{\p})|=0.998>\theta$, $|\frac{\partial f}{\partial x_3}(\tilde{\p})|=3.004>\theta$, $|\frac{\partial f}{\partial x_4}(\tilde{\p})|=4.002>\theta$, thus $f$ is $\theta$-regular at $\tilde{\p}$.
\end{exm}
From this example, it's easy to see that when compared with the exact case, the approximate zero $\tilde{\p}$ brings a small perturbation in the coefficients of the Taylor expansion of $f$ at $\tilde{\p}$. However, once given a proper $\theta$, we could acquire the same judging result as the exact case. For the above example,
 $f$ is regular at $\p$ and it is also $\theta$-regular at $\tilde{\p}$.

\begin{defin} \label{def-delta}
 Denote the operation set $\Delta\triangleq\{{+},\ {\cdot},\ {\partial}\}$, where $``\ {+}\ "$ denotes the sum of two polynomials, $``\ {\cdot}\ "$ denotes scalar multiplication and $``\ {\partial}\ "$ the differential of a polynomial. Given a polynomial system $\F=\{f_1,\ldots,f_n\}\subset \C[\x]$  and $\p\in\C^n$ such that $\F(\p)=\mathbf{0}$, we define a polynomial set $\Delta_\p(\F)$, which satisfies:
 \begin{enumerate}[$(1)$]
\item  $\F\subset \Delta_\p(\F)$;
\item  $\{a\,h|h\in\Delta_\p(\F),a\in\C\backslash\{0\}\}\subset \Delta_\p(\F)$;
\item  $\{h_1+h_2|h_1(\p)+h_2(\p)=0,h_1,h_2\in \Delta_\p(\F)\}\subset \Delta_\p(\F)$;
\item  $\{\frac{\partial h}{\partial x_i}|\frac{\partial h}{\partial x_i}(\p)=0, i\in\{1,\ldots,n\}, h\in\Delta_\p(\F)\}\subset \Delta_\p(\F)$.
\end{enumerate}
Especially, for one polynomial $f\in\C[\x]$, we have the corresponding set $\Delta_\p(f)$.
\end{defin}


The following lemma shows the relationship between the polynomials in $\Delta_\p(\F)$ and the polynomials in $\F$.
\begin{lem}\label{lem-g}
Let $\F=\{f_1,\ldots,f_n\}\subset \C[\x]$ and $\p\in\C^n$, s.t. $\F(\p)=\mathbf{0}$. $\forall g\in\Delta_\p(\F)$, we have
\begin{equation}\label{eqn-g}
g=\sum_{i=1}^n\sum_j a_{i,j}\frac{\partial^{|\boldsymbol{\gamma}_{i,j}|} f_i}{\partial \x^{\boldsymbol{\gamma}_{i,j}}},
\end{equation}
where $a_{i,j}\in\C$ and $\boldsymbol{\gamma}_{i,j}\in\N^n$.
\end{lem}
\begin{proof} The proof is obvious.
\end{proof}

We illustrate Definition \ref{def-delta} and Lemma \ref{lem-g} by the following example.
\begin{exm}
Let $\F=\{f_1=(x+y)^2+x^3,f_2=x+y+y^3\}$. $\p=(0,0)$ is an isolated zero of $\F=0$. Let $h_1=\frac{\partial f_1}{\partial x}=2\,(x+y)+3\,x^2$, $h_2=\frac{\partial f_1}{\partial y}=2\,(x+y)$, $h_3=h_1-2\ f_2=3\,x^2-2\,y^3$, $h_4=\frac{\partial h_3}{\partial x}=6\,x$, $h_5=\frac{\partial^2 h_3}{\partial y^2}=-12\,y$. It is clear that $h_i\in\Delta_\p(\F), i=1,\ldots,5$ and $h_i$ has the form as (\ref{eqn-g}).
\end{exm}

%
\section{Computing a deflated square system}
 Given a polynomial system with a multiple zero, Newton-type method usually is not used directly on the input system since it converges slowly or even doesn't converge. Thus, deflation techniques are developed to transform the input system into another deflated system, which is regular at some zero whose certain projection is the given multiple zero.

 In this section, given a polynomial system $\F\subset\C[\x]$ with an isolated singular zero $\p\in\C^n$, by employing some differential operations on the input polynomials or on the linear combinations of the related polynomials, we propose a new method to construct a new square system $\F'\subset\C[\x]$, which satisfies that $\p$ is a simple zero of $\F'=\0$. We also prove the existence of $\F'$ and show some properties of it.

 First, let's see a simple example to explain our idea. 
 \begin{exm}\label{exm-1.1}
   Let $\F=\{f_1=x-y+x^2,f_2=x-y+y^2\}$ with a 3-fold isolated zero $\p=(0,0)$. Obviously, $f_1$ and $f_2$ are already regular at $\p$. However, it is easy to find that the terms with degree one of $f_1$ and $f_2$ are linear dependent. Using $f_2-f_1$ to eliminate these terms of degree one, we get the polynomial $h=y^2-x^2$ and two new polynomials $\frac{\partial h}{\partial x}=-2\,x$, $\frac{\partial h}{\partial y}=2\,y$, which are both regular at $\p$. Selecting the two polynomials $f_1$ and $\frac{\partial h}{\partial y}$, we get a new square system $\F'=\{x-y+x^2,2\,y\}$, which has a regular zero $\p=(0,0)$. Moreover, it's a system without perturbation.
 \end{exm}
 Based on the idea in the above simple example, now we show our technique to construct a deflated square system below.

 Assume that we have got the polynomials $g_1,\ldots,g_s$, which are regular at $\p$, from the input polynomials $f_1,\ldots,f_s$ such that $$\rank(\J(g_1,\ldots,g_s)(\p))=s.$$
  Given one more polynomial $f_{s+1}$, we want to compute another polynomial $g_{s+1}$, s.t.
  \begin{equation*}
    \rank(\J(g_1,\ldots,g_s,g_{s+1})(\p))=s+1.
  \end{equation*}
  Using only $g_1,\ldots,g_s$ and $f_{s+1}$, we may not get the suitable $g_{s+1}$ if
  \begin{equation*}
    \dim\V(g_1,\ldots,g_s,f_{s+1})>\dim\V(f_1,\ldots,f_s,f_{s+1}).
  \end{equation*}

  The input polynomials are needed in this case. Thus, we use $\{g_1,\ldots,g_s\}\cup\{f_1,\ldots,f_{s+1}\}$ to compute $g_{s+1}$.  We will show how to compute $g_{s+1}$ in the following lemma.
\begin{lem}\label{lem-abc}
 Let $\F=\{f_1,\ldots,f_{s},f_{s+1},\ldots,f_{s+k}\}\subset \C[x_1,\ldots,x_n] (k\ge 1)$ and $\p \in \C^n$, s.t. $\F(\p)=\0$ and $\rank(\J(\F)(\p))=s$. Assume $\dim\V(\F)\le n-s-1$ and $\deg(\F)=m(m>1)$. Then we can get a polynomial system $\F'=\{f'_1,\ldots,f'_s,f'_{s+1}\}$, which satisfies:
\begin{enumerate}
 \item $\rank(\J(\F')(\p))=s+1$, and  $f_j'\in\Delta_\p(\F) (1\le j\le s+1)$;

 \item $\deg(\F')\le m$.
 \end{enumerate}
\end{lem}
\begin{proof}
 Without loss of generality, we assume that $\p$ is the origin and
 \begin{equation}
   \rank(\J(f_{1},\ldots,f_{s})(\p))=s.
 \end{equation}
 In the following, we consider the case of $s>0$, since if $s=0$, we can use the operators $\partial$ on $f_i(1\le i\le s+k)$ to get some polynomials, which are regular at $\p$.

  To construct a polynomial system $\F'$, s.t. $\rank(\J(\F')(\p))=s+1$, we consider the rest polynomials $\{f_{s+1},\ldots,f_{s+k}\}$.
 Our proof is constructive.

First, $f_i(i=1,\ldots,s)$ has the form : $$f_i=\sum_{k=1}^{n}a_{ik}x_{k}+ \mathrm{T}_i,$$ where $\mathrm{T}_i\in\C[\x]$ and $\deg(\mathrm{T}_i)=0$ or $\deg(\mathrm{T}_i)\geq 2$.
It's easy to know that the row vector $\mathbf{a}_i=(a_{i1},\ldots,a_{in})(1\le i\le s)$ of the Jacobian matrix of $(f_1,\ldots,f_s)$ at $\p$ are linear independent since (3) holds.


Therefore, we can consider the following linear coordinate transformation $L$:
\[
L:\left\{\begin{array}{l}
y_i=\sum\limits_{k=1}^{n}a_{ik}x_k,1\leq i\leq s\\
y_i=x_i,\ i=s+1,\ldots,n.
\end{array}\right. \ \ 
\]
 With a realignment of the sequence of the variables $\{x_1,\ldots,x_n\}$, we can always have the first $s$ columns of the coefficient matrix of $L$ being linear independent. Then $L$ is invertible. Denote the inverse of $L$ as $L^{-1}$. Let $\p'=L(\p)$ and $F_i=L^{-1}(f_i)\in\mathbb{C}[{y}_1,\ldots,{y}_n]$. We have: 

\begin{eqnarray}\label{a}
\left\{\aligned
 F_i&=y_i+L^{-1}(\mathrm{T}_i),\  i=1,\ldots,s,\\
 F_{s+i}&=\sum\limits_{j=1}^{s}b_{i,j}y_j+L^{-1}(\mathrm{T}_{s+i}),i=1,\ldots,k.
\endaligned\right.
\end{eqnarray}
%

 Since $\dim\mathbb{V}(\F)\le n-s-1$ and $L^{-1}$ is invertible, it is obvious that $$\dim\V(F_1,\ldots,F_{s+k})\le n-s-1.$$
 Therefore, noticing that the terms with degree one of all $F_i(i=1,\ldots,s+k)$ in (\ref{a}) contain only $s$ variables, there must be at least one of $\{L^{-1}(\mathrm{T}_i),i=1,\ldots,s+k\}$ containing at least one term, which has the form of $y_{s+1}^{d_{s+1}}y_{s+2}^{d_{s+2}}\cdots y_{n}^{d_{n}}$, such that $\sum\limits_{j=s+1}^{n}d_j>1$.

 It's easy to prove the claim. Suppose all $L^{-1}(\mathrm{T}_i)(1\leq i\leq s+k)$ contain no terms of the form of $y_{s+1}^{d_{s+1}}y_{s+2}^{d_{s+2}}\cdots y_{n}^{d_{n}}$.
 Then, all the terms of $F_i(1\leq i\leq s+k)$ have the form of $y_1^{d_1}\cdots y_s^{d_s}y_{s+1}^{d_{s+1}}\cdots y_{n}^{d_{n}}$, $\sum\limits_{j=1}^{s}d_j>0$.
 In this case,  the system $\{F_1,\ldots,F_{s+k}\}$ vanishes on $\{y_1=0,\ldots,y_s=0\}$. Thus, we can verify easily that $\dim\V(F_1,\ldots,F_{s+k})= n-s$, which contradicts with $\dim\V(F_1,\ldots,F_{s+k})\le n-s-1$. Thus, the claim is true.

 Without loss of generality, we suppose that $L^{-1}(\mathrm{T}_l)(l\in\{1,\ldots,s+k\})$ has the term with the form of $y_{s+1}^{d_{s+1}}y_{s+2}^{d_{s+2}}\cdots y_{n}^{d_{n}}$ and take the variable $y_{s+1}$ for example, i.e. $d_{s+1}\neq 0$. Further, we ask for the term with a lowest degree among all this
kind of terms and denote the lowest degree as ${d}$.
 Then, we have:
 \begin{eqnarray}
 \aligned
   F'_{s+1}&=\frac{\partial^{{d}-1} F_{l}}{\partial y_{s+1}^{d_{s+1}-1}y_{s+2}^{d_{s+2}}\cdots y_{n}^{d_{n}}}=\sum\limits_{i=1}^{n}\gamma_i y_i+\mathrm{T}'_{l},\ {d}=\sum\limits_{j=s+1}^{n}d_j.
   \endaligned
 \end{eqnarray}
 It is easy to see that $\gamma_{s+1}\neq0,\ \deg(F_{s+1}')<\deg(F_l)$.

 Thus, we have a new system $\{F_1,\ldots,F_s,F'_{s+1}\}$. It's easy to check that $$\rank(\J(F_1,\ldots,F_s,F'_{s+1})(\p'))=s+1.$$
 Finally, after doing the transformation $L$ on $F_i(1\leq i\leq s)$ and $F'_{s+1}$, we have the new system $\F'=\{f'_1,\ldots,f'_{s+1}\}$, where $$f'_i=L(F_i)=f_i(i=1,\ldots,s),f'_{s+1}=L(F'_{s+1})\ \text{with}\ \rank(\J(f'_1,\ldots,f'_s,f'_{s+1})(\p))=s+1.$$
 By the definition of $\Delta_\p(\F)$ (see Definition \ref{def-delta}), we can find that $f_i'\in\Delta_\p(\F)(1\le i\le s+1)$. Therefore, we finished the first part of the proof.

From Lemma \ref{lem-g} and (5), it is easy to know that the maximal degree of $f'_i(i=1,\ldots,s+1)$ is no larger than $m$. That is, $\deg(\F')\le m$. Thus, we complete the proof.
\end{proof}

Now, we consider constructing a square system, which is regular at an isolated singular zero of the input system.

\begin{thm}\label{thm-theory}
Let $\F=\{f_1,\ldots,f_N\}\subset \C[\x](N\ge n)$ be a polynomial system. $\p\in\C^n$ an isolated singular zero of $\F=\0$ and $\deg(\F)=m$. Then there exists a square polynomial system $\F'=\{f'_1,\ldots,f'_n\}\subset\Delta_\p(\F)$, s.t.
\begin{enumerate}
\item $\p$ is an isolated regular zero of $\F'=\0$;
\item $\deg(\F')\le m$.
\end{enumerate}
\end{thm}
\begin{proof}

 Without loss of generality, assume that $\p$ is the origin. In the following,
we will construct a square system by the polynomials in $\Delta_\p(\F)$.

First, we can choose a system $\F_{0}$ from $\F$, denoted as $\F_0= \{f_1,\ldots,f_r\}$, whose Jacobian matrix at $\p$ has a maximal rank, s.t.  $$\rank(\J(f_{1},\ldots,f_{r})(\p))=\rank(\J(\F)(\p))=r,\ 0\le r\le n.$$ If $r=n$, we finish the proof.
Noticing that when $r=0$, we need only considering at least one of the polynomials in $f_1,\ldots,f_N$ and can always get at least one polynomial, which is regular at $\p$ by Lemma \ref{lem-partial}. Thus, in the following, we consider the case of $1\le r<n$.

First, considering the system $\{f_1,\ldots,f_r,f_{r+1},\ldots,f_N\}$, by Lemma \ref{lem-abc}, we can get a system $$\F_1=\{f_1^{(1)},\ldots,f_r^{(1)},f_{r+1}^{(1)}\},$$ s.t. $$\F_1(\p)=\0 \ \text{and} \ \rank(\J(f_1^{(1)},\ldots,f_r^{(1)},f_{r+1}^{(1)})(\p))=r+1.$$
Using the technique in Lemma \ref{lem-abc}, when considering the system $\F \cup\{f_1^{(1)},\ldots,f_r^{(1)},f_{r+1}^{(1)}\}$, we can get a system $$\F_2=\{f_1^{(2)},\ldots,f_{r+1}^{(2)},f_{r+2}^{(2)}\},$$ s.t. $$\F_2(\p)=\0 \ \text{and} \ \rank(\J(f_1^{(2)},\ldots,f_{r+1}^{(2)},f_{r+2}^{(2)})(\p))=r+2.$$
Repeat this process $n-r$ times and finally, we get a square system $$\F_{n-r}=\{f_1^{(n-r)},f_{2}^{(n-r)},\ldots,f_{n}^{(n-r)}\},$$ s.t. $$\F_{n-r}(\p)=\0 \ \text{and} \ \rank(\J(f_1^{(n-r)},f_{2}^{(n-r)},\ldots,f_{n}^{(n-r)})(\p))=n.$$
 Thus, our final square system $$\F'=\{f'_1=f_1^{(n-r)},f'_2=f_{2}^{(n-r)},\ldots,f'_n=f_{n}^{(n-r)}\}.$$
By Lemma \ref{lem-abc}, it is obvious that the maximal degree of $f'_i(1\leq i\leq n)$ is no larger than $m$. That is, $\deg(\F')\le m$. 
%
%
%
%
\end{proof}
\begin{rem}

  1. In the above construction process, we repeat $n-r$ times to get the deflated system $\F'$. If considering all the variables simultaneously, we get more than one eligible polynomial each time in (5). Thus, the number of times in actual computation is less than $n-r$.

  2. In the beginning of our construction, we also can compute all the related polynomials of all the input polynomials, which are regular at $\p$. Then, we choose a system from these polynomials, whose Jacobian matrix at $\p$ has a maximal rank. That's to say that we make $r$ as big as possible to reduce our repeating steps.
\end{rem}
%
Theorem \ref{thm-theory} tells us that given a polynomial system $\F$ with an isolated singular zero $\p$, we can construct a new square system $\F'$, which is regular at $\p$ and moreover, the degree of the polynomials in $\F'$ does not increase. We give an example to illustrate our method.
\begin{exm}\label{exm-DZ2}
 (DZ2 \cite{dayton1}) Let $\F=\{f_1=x_1^4,f_2=x_1^2x_2+x_2^4,f_3=x_3+x_3^2-7\,x_1^3-8\,x_1^2\}$, which has a 16-fold zero $\p=(0,0,-1)$. The maximal degree of $f_1,f_2,f_3$ is 4. First, by the Taylor expansions of $f_1,f_2,f_3$ at $\p$, we have:
  \begin{eqnarray*}
  \aligned
    f_1 &= x_1^4, \\
    f_2 &= x_1^2x_2+x_2^4, \\
    f_3 &= -(x_3+1)-8\,x_1^2+(x_3+1)^2-7\,x_1^3.
    \endaligned
  \end{eqnarray*}
  It's easy to find that only $f_3$ is regular at $\p$. Since $s=\rank(\J(\F)(\p))=1$ and $\dim\V(f_3,f_2)=1$, we consider the system $\{f_3,f_2\}$ directly. By Lemma \ref{lem-abc}, we have a system $$\{f_1^{(1)}=f_3,f_2^{(1)}=\mathbf{d}^{(2,0,0)}_{\x}(f_2)=x_2\},$$ which satisfies $\rank(\J(f_1^{(1)},f_2^{(1)})(\p))=2$.

  Next, we continue to consider the system $\{f_1^{(1)},f_2^{(1)}\}\cup\F$. Since $\dim\V(f_1^{(1)},f_2^{(1)},\F)=0$, by Lemma \ref{lem-abc}, we have a system $$\{f_1^{(2)}=f_3,f_2^{(2)}=x_2,f_3^{(2)}=\mathbf{d}^{(3,0,0)}_{\x}(f_1)=4\,x_1\},$$ which satisfies $\rank(\J(f_1^{(2)},f_2^{(2)},f_3^{(2)})(\p))=3$.

  Thus, we acquire the final square system $\F'=\{f_3,x_2,4\,x_1\}$. It's easy to check that $\p$ is a simple zero of $\F'=\0$ and the degree of every polynomial in $\F'$ is no more than 4.
\end{exm}
In this example, we repeat $n-s=2$ times to acquire the final square system $\F'$. In fact, as what we say in Remark 2 of Theorem \ref{thm-theory}, computing twice is not necessary. Noticing that when computing $f_2^{(1)}=\mathbf{d}^{(2,0,0)}_{\x}(f_2)=x_2$, we also can get $\mathbf{d}^{(1,1,0)}_{\x}(f_2)=2\,x_1$. They are both regular at $\p$. It is easy to check that $$\rank(\J(f_1^{(1)},f_2^{(1)},\mathbf{d}^{(1,1,0)}_{\x}(f_2)=2\,x_1)(\p))=3.$$
Thus, we obtain another square system $\F'=\{f_3,x_2,2\,x_1\}$.

\section{An effective version of our deflation method}

In the section, by introducing some new variables to represent the coefficients of the linear combinations, we give an effective version of our deflation method.
The deflated system produced by our deflation method has a simple zero, whose partial projection corresponds to the isolated singular zero of the input system.  Furthermore, we also analyze the influences of the given tolerances $\theta$ and $\varepsilon$ to our method and show how to adjust their values to get a deflated system as exact as possible.


\subsection{Parametric deflation system}
Given a polynomial system $\F$ with an isolated singular zero $\p$, by employing some differential operations on the input polynomials directly or on the linear combinations of the related polynomials, we give a method to construct a new polynomial system $\F'$ in Section 3, which satisfies that $\p$ is a simple zero of $\F'=\0$.

However, in practice, we can just get an approximate zero $\tilde{\p}$. As what we say in Example 1, the inexact value of $\tilde{\p}$ usually brings perturbations in the coefficients when doing the Taylor series expansions of the input polynomials at $\tilde{\p}$. Therefore, we can not do exact computations when adding two or more polynomials together. The inexact computations would produce a perturbed system of $\F'$, which will lead to a bad final deflation result. We show an example to illustrate this case.
\begin{exm}\label{exm-1.2}
  Continue with Example \ref{exm-1.1}. Given an approximate zero $$\tilde{\p}=(0.0006721,0.0008381).$$
  Using the method in Theorem \ref{thm-theory}, we have $\tilde{h}=f_2+\tilde{\alpha} f_1$. By solving a Least Square problem, we can get $\tilde{\alpha}=-0.9984909264232$.
  Finally, we get an inexact system $$\widetilde{\F}'=\{x-y+x^2,2\,y-0.0015090735767\}.$$
  Obviously, we can not get a good result by the system $\widetilde{\F}'$.
\end{exm}
With a simple analysis, we can find that we couldn't get an exact coefficient $\alpha$ of the linear combination of the polynomials with an approximate zero.


In the following, by introducing some new variables to represent the coefficients of the linear combinations, we give an effective version of our deflation method. Finally, the effective version of our deflation method will usually produce an exact deflated system, which has a simple zero, whose partial projection corresponds to the isolated singular zero of the input system. Furthermore, we also provide the size bound of our method. To our knowledge, it is the first time that considering the deflation of the polynomial system from the perspective of linear combination.

 Similarly, before giving our theoretical results, we also show our main idea with a simple example first.
 \begin{exm}
 Still consider Example \ref{exm-1.1}. Once given an approximate zero of the input system: $\tilde{\p}=(0.0006721,0.0008381)$, by Example \ref{exm-1.2}, we know the coefficient $\tilde{\alpha}$ is inexact. Now we introduce a new variable $\alpha_1$. Let $h=f_2+\alpha_1 f_1$ and compute $$\frac{\partial h}{\partial x}=1+\alpha_1(2\,x+1),\frac{\partial h}{\partial y}=2\,y-1-\alpha_1.$$
 Similar as in Example \ref{exm-1.2}, we have $\tilde{\alphahei}_1=-0.9984909264232$. Given a tolerance $\varepsilon=0.05$, we have $$\rank(\J(f_1,\frac{\partial h}{\partial x},\frac{\partial h}{\partial y})(\tilde{\p},\tilde{\alpha}_1),\varepsilon)=2<3.$$
 Do once again this process and introduce two new variables $\alpha_2,\alpha_3$. Let $$g=\frac{\partial h}{\partial y}+\alpha_2 f_1+\alpha_3\frac{\partial h}{\partial x}$$ and compute
 $$\frac{\partial g}{\partial x}=2\,\alpha_1 \alpha_3+\alpha_2(2\,x+1),\frac{\partial g}{\partial y}=2-\alpha_2,\frac{\partial g}{\partial \alpha_1}=\alpha_3(2\,x+1)-1.$$
 By solving another Least Square problem, we get the approximate values: $$\tilde{\alpha}_2=1.998595\\5412653,\tilde{\alpha}_3=1.0014510032456.$$ Then, we have $$\rank(\J(f_1,\frac{\partial f}{\partial x},\frac{\partial g}{\partial x},\frac{\partial g}{\partial y},\frac{\partial g}{\partial \alpha_1})(\tilde{\p},\tilde{\alpha}_1,\tilde{\alpha}_2,\tilde{\alpha}_3),\varepsilon)=5.$$
 Thus, we get a polynomial system $$\widetilde{\F}'(\x,\alphahei)=\{f_1,\frac{\partial h}{\partial x},\frac{\partial g}{\partial x},\frac{\partial g}{\partial y},\frac{\partial g}{\partial \alpha_1}\},$$ whose Jacobian matrix at $(\tilde{\p},\tilde{\alpha}_1,\tilde{\alpha}_2,\tilde{\alpha}_3)$ has a full rank under the tolerance $\varepsilon$. 
 In fact, we can find that $(0,0,-1,2,1)$ is a simple zero of $\widetilde{\F}'(\x,\alphahei)=\0$ and the partial projection $(0,0)$ of $(0,0,-1,2,1)$ corresponds to the isolated singular zero $\p$ of the input system $\F$.
 \end{exm}
Given a polynomial system with an isolated zero, we have the following lemma.
\begin{lem}\cite{lvz}\label{lem_lvz}
  Let $\F=\{f_1,\ldots,f_n\}\subset\C[\x]$ be a polynomial system. $\p\in\C^n$ is an isolated singular zero of $\F=\0$. $\llambda=(\lambda_1,\ldots,\lambda_n)\in\C^n$ is a nonzero row vector, which satisfies $\J(\F)(\p)\llambda^{\T}=\0$. For the new system $$\G=\{\lambda_1\frac{\partial f_1}{\partial x_1}+\ldots+\lambda_n\frac{\partial f_1}{\partial x_n},\ \ldots,\ \lambda_1\frac{\partial f_n}{\partial x_1}+\ldots+\lambda_n\frac{\partial f_n}{\partial x_n}\},$$ we have the multiplicity of $\p$ in $\{\F,\G\}=\0$ is lower than the multiplicity of $\p$ in $\F=\0$.
\end{lem}

\begin{rem} In Remark 2.1 of \cite{hauensten}, the authors mentioned that deflation could also be constructed using the left null space. That is, we can replace $\G$ by the following system
\begin{equation}\label{eq-G}
\G'=\{\lambda_1\frac{\partial f_1}{\partial x_1}+\ldots+\lambda_n\frac{\partial f_n}{\partial x_1},\ \ldots,\ \lambda_1\frac{\partial f_1}{\partial x_n}+\ldots+\lambda_n\frac{\partial f_n}{\partial x_n}\},
\end{equation}
where $\llambda\,\J(\F)(\p)=\0$. Furthermore, we have the following lemma.
\end{rem}

\begin{lem}\label{lem-s}
  Let $\F=\{f_1,\ldots,f_n\}\subset\C[\x]$ be a polynomial system. $\p\in\C^n$ be an isolated singular zero of $\F=\0$. Assume $\rank(\J(f_1,\ldots,f_s)(\p))=\rank(\J(\F)(\p))=s$. Consider the augmented system $$\G=\{f_1,\ldots,f_n,h_1,\ldots,h_n\}\subset\C[\x,\alphahei],$$ where $$h_j=\alpha_1\frac{\partial f_1}{\partial x_j}+\cdots+\alpha_s\frac{\partial f_s}{\partial x_j}+\frac{\partial f_{s+1}}{\partial x_j},\ j=1,\ldots,n.$$
  Then, we have:
  \begin{enumerate}
    \item there exists a unique $\hat{\alphahei}\in\C^s$ such that the system $\G$ has an isolated zero at $(\p,\hat{\alphahei})$.
    \item the multiplicity of  $\G$ at $(\p,\hat{\alphahei})$ is lower than that of $\F$ at $\p$.
  \end{enumerate}

\end{lem}
\begin{proof}
Let $$A_{ij}(\x)=\frac{\partial f_i}{\partial x_j}\in\C[\x],\ a_{ij}=\frac{\partial f_i(\p)}{\partial x_j}\in\C,\ i=1,\ldots,s+1,\ j=1,\ldots,n.$$
Denote the matrix $A=(a_{ij}),i=1,\ldots,s,j=1,\ldots,n$ and the row vector $b=(a_{s+1,1},\ldots,a_{s+1,n})$.

  On one hand, when we fix $\x=\p$, the system $$\H(\p,\alphahei)=\{h_j(\p,\alphahei)=a_{1j}\alpha_1+\ldots+a_{sj}\alpha_s+a_{s+1,j},\ j=1,\ldots,n\}$$ is a linear system with respect to the variables $\alpha_1,\ldots,\alpha_s$. Furthermore, it is easy to check that $\hat{\alphahei}$, which is determined by $AA^\T\hat{\alphahei}=-Ab^\T$, is the unique zero of $\H(\p,\alphahei)=\0$. That is, there exists a unique $\hat{\alphahei}$ such that the system $\G$ has an isolated zero at $(\p,\hat{\alphahei})$.



  On the other hand, with the row operations, we could reduce the system $\G$ to the system $$\{\alpha_1=l_1(\x),\ldots,\alpha_s=l_s(\x)\},$$ where $l_i(\x)$ are rational expressions and $\hat{\alpha}_i=l_i(\p)$. Thus, considering the multiplicity of $\G$ at $(\p,\hat{\alphahei})$ is equivalent to considering the multiplicity of $\G(\x,\hat{\alphahei})$ at $\p$. Note that $(\hat{\alpha}_1,\ldots,\hat{\alpha}_s,1,0,\ldots,0)\,\J(\F)(\p)=\0$. By Lemma \ref{lem_lvz} and (\ref{eq-G}), we know the second part holds. Thus, we finished the proof.
\end{proof}

%
%
%
%
In the above lemma, we construct $n$ new polynomials $h_1,\ldots,h_n$. In fact, we can get them from the following way. Note that $$\rank(\J(f_1,\ldots,f_s)(\p))=\rank(\J(\F)(\p))=s.$$ We know easily that $\J(f_{s+1})(\p)$ and $\J(f_{1})(\p),\ldots,\J(f_{s})(\p)$ are linearly dependent. Thus,
 we can do the linear combination between $f_{s+1}$ and $f_1,\ldots,f_s$ to eliminate this linear relationship. Let
 \begin{equation}
   g=f_{s+1}+\sum_{i=1}^{s}\alpha_if_i,
 \end{equation}
 where new variables $\alpha_i$ are used to represent the coefficients of the linear combination. Compute all the derivatives of $g$ with respect to the variables $x_1,\ldots,x_n$ and we get $$h_j=\frac{\partial g}{\partial x_j}=\alpha_1\frac{\partial f_1}{\partial x_j}+\cdots+\alpha_s\frac{\partial f_s}{\partial x_j}+\frac{\partial f_{s+1}}{\partial x_j},\ j=1,\ldots,n.$$

Thus, the above lemma tells us that after doing the linear combination of polynomials between $f_{s+1}$ and $f_1,\ldots,f_s$, we get an augmented system $\G$, which satisfies that the multiplicity of  $\G$ at $(\p,\hat{\alphahei})$ is lower than that of $\F$ at $\p$.
By repeating using the linear combination between polynomials in the original system and its related derivatives, we can construct a final deflated system, which processes an isolated simple zero. Denote $\mu$ be the multiplicity of $\F$ at $\p$. We do this repetitive process at most $\mu$ times.

Further, based on Lemma \ref{lem-s}, we have the following theorem.
\begin{thm}\label{thm-main}
  Let $\F=\{f_1,\ldots,f_n\}\subset\C[\x]$ be a polynomial system. $\p\in\C^n$ be an isolated singular zero of $\F=\0$. Denote $m=\deg(\F)$. Then there exists a square polynomial system $\widetilde{\F}'(\x,\alphahei)=\{g_1,\ldots,g_t\}\subset\C[\x,\alphahei]$, s.t.
\begin{enumerate}
  \item 
  $({\p},\hat{\alphahei})\in\C^t$ is an isolated simple zero of $\widetilde{\F}'(\x,\alphahei)=\0$;
  \item $t$ is bounded by $2^\mu\,n$, where $\mu$ is the multiplicity of $\p$ in $\F$;
  \item $\deg(\widetilde{\F}'(\x,\alphahei))\le m$.
\end{enumerate}
\end{thm}

Next, based on Lemma \ref{lem-s} and Theorem \ref{thm-main}, we give an effective algorithm to compute a deflated square system from the input system with an approximate isolated singular zero below. It is an effective version of Lemma \ref{lem-abc}. $\theta$ is a tolerance to detect the regularity of the polynomials and we will talk about it in next subsection. $\varepsilon$ is another tolerance to judge the numerical rank of the Jacobian matrix at an approximate zero and we also talk about it in next subsection.

 \begin{algorithm}[H]         
\caption{ $\mathbf{CDSS}$ : Compute a deflated square system.}             
\label{alg-subalg}                  

\begin{algorithmic}[1]               

\REQUIRE ~~                          

a polynomial system $\F:=\{f_1,\ldots,f_n\}\subset \C[\x]$, an approximate isolated singular solution $\tilde{\p}\in\C^n$, two tolerances $\theta$ and $\varepsilon$.

\ENSURE ~~                           

    a square polynomial system $\widetilde{\F}'(\x,\alphahei):=\{\tilde{f}_1,\ldots,\tilde{f}_t\}\subset\C[\x,\alphahei]$ and a point $\tilde{\alphahei}$, s.t. $(\tilde{\p},\tilde{\alphahei})$ is an approximate regular zero of $\widetilde{\F}'(\x,\alphahei)=\0$. 

\STATE Compute $ \G=\{\mathbf{d}^{\boldsymbol{\gamma}}_\mathbf{x}(f)| \mathbf{d}^{\boldsymbol{\gamma}}_\mathbf{x}(f)$ is $\theta$-regular at $\tilde{\p}, f\in\F \}$;
\STATE Let $\H:=\F\cup\G$, $\X:=\x$;

\WHILE {$\rank(\J(\H)(\tilde{\p}),\varepsilon)\neq |\X|$}

\STATE Compute $r:=\rank(\J(\H)(\tilde{\p}),\varepsilon)$;

\STATE  Choose any $\H_1:=\{h_1,\ldots,h_{r}\}\subset\H$, s.t. $\rank(\J(\H_1)(\tilde{\p}),\varepsilon)=r$;


\STATE Choose $h_{r+1}:=\H\setminus\H_1$, s.t. $\dim \V(\H_1,h_{r+1})=n-r-1$;

\STATE Let $g:=h_{r+1}+\sum\limits_{j=1}^{r}\alpha_j h_j$;

\STATE Compute $\tilde{\alphahei}:=LeastSquares((\J(\H_1,h_{r+1})(\tilde{\p}))^{\T}(\alphahei,1)^{\T}=\0)$;

\STATE Compute $g_{1}:=\J_1(g),\ldots,g_{n}:=\J_n(g)$;

\STATE Set $\H:=\{\H,g_{1},\ldots,g_{n}\}$, $\X:=\x\cup\alphahei$ and $\tilde{\p}:=(\tilde{\p},\tilde{\alphahei})$;

\ENDWHILE

\RETURN  a square system $\widetilde{\F}'(\x,{\alphahei})=\{\H_1,g_1,\ldots,g_n\}$ and a point $\tilde{\alphahei}$.\\

\end{algorithmic}
\end{algorithm}

\begin{rem}
   1. The termination and correctness of the algorithm is guaranteed by Lemma \ref{lem-s} and Theorem \ref{thm-main}.

   2. In the above algorithm, we compute polynomials of every $f_i$, which are regular at $\p$ at the beginning. Then, we put all these polynomials together to compute a system $\F_0$, such that the rank of its Jacobian matrix at $\p$ is maximal. This operation can make $r$ as big as possible. In some cases, we have $r=n$, which means we need not introduce new variables, such as Example \ref{exm-2}. The aim of this preprocessing operation can speed up our algorithm.

\end{rem}


Now, we give two examples to illustrate Algorithm 1.
\begin{exm}\label{exm-4}
Consider a polynomial system $\F=\{f_1=-\frac{9}{4}+ \frac{3}{2}\,x_1+2\,x_2+3\,x_3+4\,x_4-\frac{1}{4}\,x_1^2,f_2=x_1-2\,x_2-2\,x_3-4\,x_4+2\,x_1x_2+3\,x_1x_3+4\,x_1x_4,f_3=8-4\,x_1-8\,x_4+2\,x_4^2+4\,x_1x_4-x_1x_4^2,
f_4=-3+3\,x_1+2\,x_2+4\,x_3+4\,x_4\}$. Given an approximate singular zero $$\tilde{\p}=(\tilde{p}_1,\tilde{p}_2,\tilde{p}_3,\tilde{p}_4)=(1.00004659,-1.99995813,-0.99991547,2.00005261)$$ of $\F=\0$ and the tolerance $\varepsilon=0.005$.

 First, we have the Taylor expansion of $f_3$ at $\tilde{\p}$: $$f_3=3\cdot10^{-9}-3\cdot10^{-9}(x_1-\tilde{p}_1)+0.00010522(x_4-\tilde{p}_4)+0.99995341(x_4-\tilde{p}_4)^2$$
 $$-0.00010522(x_1-\tilde{p}_1)(x_4-\tilde{p}_4)-(x_1-\tilde{p}_1)(x_4-\tilde{p}_4)^2.$$
%
Consider the tolerance $\theta=0.05$. Since $$|f_3(\tilde{\p})|<\theta,\ \left|\frac{\partial f_3}{\partial x_i}(\tilde{\p})\right|<\theta(i=1,2,3,4),\  \left|\frac{\partial^2 f_3}{\partial x_4^2}(\tilde{\p})\right|>\theta,$$
 we get a polynomial $$\frac{\partial f_3}{\partial x_4}=-8+4\,x_1+4\,x_4-2\,x_1x_4,$$ which is $\theta$-regular at $\tilde{\p}$.
 Similarly, by the Taylor expansion of $f_1,f_2,f_4$ at $\tilde{\p}$, we have that $f_1,f_2,f_4$ are all $\theta$-regular at $\tilde{\p}$.

 Thus, by Algorithm \ref{alg-subalg}, we have $\G=\{f_1,f_2,-8+4\,x_1+4\,x_4-2\,x_1x_4,f_4\}$. Compute $$r=\rank(\J(\G)(\tilde{\p}),\varepsilon)=3.$$ We can choose $\H_1=\{h_1=f_1,h_2=f_2,h_3=-8+4\,x_1+4\,x_4-2\,x_1x_4\}$ from $\H=\G\cup\F$.
 To $h_4=f_4\in\H\setminus\H_1$, let $$g=h_4+\alpha_1 h_1 +\alpha_2 h_2 +\alpha_3 h_3.$$

 First, by solving a Least Square problem: $$LeastSquares((\J(\H_1,h_4)(\tilde{\p}))^{T}[\alpha_1,\alpha_2,\alpha_3,1]^{T}=0),$$
 we get an approximate value: $$(\tilde{\alpha}_1,\tilde{\alpha}_2,\tilde{\alpha}_3)=(-1.000006509,-0.9997557989,0.000106178711).$$

 Then, compute 
 \begin{eqnarray*}
\left\{\aligned
   g_1=\frac{\partial g}{\partial x_1}&=3+\frac{3}{2}\alpha_1+\alpha_2+4\alpha_3-\frac{1}{2}\alpha_1x_1+2\alpha_2x_2+3\alpha_2x_3+4\alpha_2x_4-2\alpha_3x_4,\\
   g_2=\frac{\partial g}{\partial x_2}&=2+2\alpha_1-2\alpha_2+2\alpha_2x_1,  \\
   g_3=\frac{\partial g}{\partial x_3}&=4+3\alpha_1-2\alpha_2+3\alpha_2x_1,  \\
   g_4=\frac{\partial g}{\partial x_4}&=4+4\alpha_1-4\alpha_2+4\alpha_3+4\alpha_2x_1-2\alpha_3x_1,
  \endaligned\right.
\end{eqnarray*}
 and we get the polynomial set $$\mathbf{H}'=\{h_1,h_2,h_3,g_1,g_2,g_3,g_4\},$$ which satisfies $$\rank(\J(\mathbf{H}')(\tilde{\p},\tilde{\alpha}_1,\tilde{\alpha}_2,\tilde{\alpha}_3),\varepsilon)=7.$$
 Thus, we get the final square system $\widetilde{\F}'(\x,\alphahei)=\mathbf{H}_1$ and the point $\tilde{\alphahei}=(\tilde{\alpha}_1,\tilde{\alpha}_2,\tilde{\alpha}_3)=(-1.000006509,-0.9997557989,0.000106178711)$.
\end{exm}

 In this example, given the input polynomial system $\F$ with an approximate singular zero $\tilde{\p}$, we can get a final square system by Algorithm 1 with only one step. In fact, $\alpha_3$ is not necessary to be introduced in this example by noticing that we can acquire a needed square system $\widetilde{\F}'(\x,\alphahei)$ by using $F=f_4+\alpha_1 f_1 +\alpha_2 f_2$. We give another example to illustrate the case that we do not introduce new variables.

\begin{exm}\label{exm-2}
(DZ2) Continue with Example \ref{exm-DZ2}.
Given an approximate isolated singular zero $$\tilde{\p}=(\tilde{p}_1,\tilde{p}_2,\tilde{p}_3)=(0.00006787,0.00007577,-0.9999)$$ and a tolerance $\varepsilon=0.005$, we use the Taylor series to expand $f_i(i=1,2,3)$ at $\tilde{\p}$ and compare all the coefficients with a tolerance $\theta=\varepsilon$. For $f_1$, we have
{\small
\begin{eqnarray*}
  f_1&=2.121833963630161\cdot 10^{-17}+1.250528341612\cdot10^{-12}(x_1-\tilde{p}_1)+2.76380214\cdot10^{-8}\\
    &(x_1-\tilde{p}_1)^2+0.27148\cdot10^{-3}(x_1-\tilde{p}_1)^3+(x_1-\tilde{p}_1)^4.
\end{eqnarray*}}
It is obvious that only the absolute value of the coefficient of $(x_1-\tilde{p}_1)^4$ is bigger than $\theta$. Therefore, compute $\mathbf{d}^{(3,0,0)}_{(x_1,x_2,x_3)}(f_1)=4x$, which is $\theta$-regular at $\tilde{\p}$. Similarly, for $f_2,f_3$, we have the corresponding polynomial(s): $\{2x_1,x_2\}\ \text{and}\ f_3.$
Thus, we have $\G=\{4\,x_1,2\,x_1,x_2,f_3\}$.
It is easy to check that $$r=\rank(\J(\G)(\tilde{\p}),\varepsilon)=\rank(\J(4x_1,x_2,f_3)(\tilde{\p}),\varepsilon)=3.$$
Thus, we get the needed square system $\widetilde{\F}'(\x)=\G=\{4x_1,x_2,f_3\}$.
\end{exm}

In the above two examples, we assume that we have a right judgement on the tolerances $\theta$ and $\varepsilon$. In fact, the choice of the tolerances $\theta$ and $\varepsilon$ is important to our algorithm. Next, we give some analysis of them.

\subsection{The analysis of \texorpdfstring{$\theta$}{theta} and \texorpdfstring{$\varepsilon$}{varepsilon}}

As what we say in Example 1, $\theta$ is an important parameter in deciding if a polynomial is $\theta$-regular at $\tilde{\p}$. The other important parameter involved in our actual computation is $\varepsilon$, which is used to judge the numerical rank of the Jacobian matrix. Therefore, in this section, we will give some analysis about the parameters $\theta$ and $\varepsilon$.

First, we point out that $\theta$ is related to the absolute values of the coefficients of the Taylor expansion of the polynomial at its approximate zero.

For example, given a polynomial $f=x^2+10000\,y^2$ with an approximate zero $\tilde{\p}=(\tilde{p}_1,\tilde{p}_2)=(0.0006851,-0.0004368)$, we have the Taylor expansion of $f$ at $\tilde{\p}$: $$f=0.001908411762+0.0013702(x-\tilde{p}_1)-8.7360(y-\tilde{p}_2)+(x-\tilde{p}_1)^2+10000(y-\tilde{p}_2)^2.$$
Given $\theta=0.5$, we have $$|f(\tilde{\p})|<\theta,|\frac{\partial f}{\partial x}(\tilde{\p})|<\theta,|\frac{\partial f}{\partial y}(\tilde{\p})|>\theta.$$ Thus, we draw the conclusion that $f$ is $\theta$-regular at $\tilde{\p}$. However, considering that the lowest degree of $f$ is 2, we know that $f$ is singular at the exact zero $\p=(0,0)$ actually, which is a different result from the case of $\tilde{\p}$.
That means $\theta$ is not chosen properly. The main reason is that the coefficient of $f$ has a great fluctuation or the accuracy of $\tilde{\p}$ is not high enough. If given another approximate zero $\tilde{\q}=(\tilde{q}_1,\tilde{q}_2)=(0.000006851,-0.000004368)$ with higher precision, we have:
$$f=1.908411762\cdot10^{-7}+0.000013702(x-\tilde{q}_1)-0.087360(y-\tilde{q}_2)+(x-\tilde{q}_1)^2+10000(y-\tilde{q}_2)^2.$$
By this time, using the same $\theta=0.5$, we have $f$ is $\theta$-singular at $\tilde{\q}$, which is the same judgement as the exact case of $\p$.

In actual computation, to deal with this case, we give one solution: For a nonzero polynomial $f\in\C[\x]$, let $\Gamma_f$ be a set of the absolute values of all the coefficients of $f$. We denote the maximal and minimal ones inside $\Gamma_f$ as $M=\max(\Gamma_f)$ and $m=\min(\Gamma_f)$ respectively. If $m/M\le10^{-a}$, we regard that the coefficients of $f$ fluctuate a lot and take $\epsilon=(m+M)/2M$; Else, we take $\theta=(m+M)/(2M\times10^a)$, where $a\in \mathbb{N}$ is related to the precision of the given approximate zero $\tilde{\p}$. For example, if the accuracy of the given approximate zero $\tilde{\p}$ has three significant digits, we can take $a=3$.
Of course, we can overcome this problem thoroughly by refining the approximate zero to a higher precision with the input system if the Jacobian matrix of the system at $\tilde{\p}$ is numerically nonsingular.

In summary, the reason for the above situation is that we judge a polynomial, which is singular at the exact zero $\p$, as a polynomial being $\theta$-regular at the approximate zero $\tilde{\p}$.

 The other situation is that a polynomial, which is regular at the exact zero $\p$, may be judged as a polynomial being $\theta$-singular at the approximate zero $\tilde{\p}$.

 For example, consider the polynomial $f=\frac{1}{20}x+x^2+10000y^2$ with the approximate zero $\tilde{\q}=(\tilde{q}_1,\tilde{q}_2)=(0.000006851,-0.000004368)$. We have:
$$f=5.333911762\cdot10^{-7}+0.05001370(x-\tilde{q}_1)-0.087360(y-\tilde{q}_2)+(x-\tilde{q}_1)^2+10000(y-\tilde{q}_2)^2.$$
Still use $\theta=0.5$ and we get the judgement that $f$ is $\theta$-singular at $\tilde{\q}$. In fact, $f$ is regular at $\p=(0,0)$. One way to deal with this case is that we can take a smaller $\theta$. When we take $\theta=0.05$, we will acquire the appropriate result.

From the above analysis about the tolerance $\theta$, we know that the choice of $\theta$ is crucial to our method. We give a further theoretical analysis about the tolerance $\theta$ below. Here, we assume that the judgement of the other tolerance $\varepsilon$, which is used to decide the numerical rank of the Jacobian matrix at the approximate zero, is correct.

Let $\theta$ be a tolerance. Assume that we have computed an intermediate system $\H=\{h_1,\ldots,h_s\}\subset\C[\x']$. Denote $\x'=(\x,\alphahei)$. Assume that $\p$ is an isolated singular zero of the original system. The exact value of $\alphahei$ related to the coefficients of linear combinations is $\hat{\alphahei}$. Denote $\p'=(\p,\hat{\alphahei})$. Let $\tilde{\p}'$ be an approximate zero of $\H$ related to $\p'$ such that
$$\rank(\J(\H)(\tilde{\p}'))=s.$$
Next, we consider one more polynomial $h\in\C[\x']$. If $h$, which is regular at $\p'$, is judged as being $\theta$-singular at $\tilde{\p}'$, we may get a perturbed system finally. Specifically, compute the Taylor expansion of $h$ at $\tilde{\p}'$:
$$h=h(\tilde{\p}')+\sum\limits_{j}^{ }\frac{\partial h(\tilde{\p}')}{\partial x_j}(x_j-\tilde{p}'_j)+\sum\limits_{i,j}^{ }\frac{\partial h^2(\tilde{\p}')}{\partial x_i \partial x_j}(x_i-\tilde{p}'_i)(x_j-\tilde{p}'_j)+\cdots.$$

Since $h$ is $\theta$-singular at $\tilde{\p}'$, we know that $|h(\tilde{\p}')|<\theta$ and all $|\frac{\partial h(\tilde{\p}')}{\partial x_j}|<\theta$.
Thus, we compute
\begin{equation}\label{eq_h}
 \frac{\partial h}{\partial x_j}=\frac{\partial h(\tilde{\p}')}{\partial x_j}+2\sum_{i}^{ }\frac{\partial^2 h(\tilde{\p}')}{\partial x_i\partial x_j}(x_i-\tilde{\p}'_i)+\cdots.
\end{equation}

If there exists some $j$ such that
$$\rank(\J(\H,\frac{\partial h}{\partial x_j})(\tilde{\p}'))=s+1 \text{ and } \frac{\partial h(\p')}{\partial x_j}\neq 0,$$
we may derive a perturbed system in the end, where $\frac{\partial h}{\partial x_j}$ has and only has one perturbed term $\frac{\partial h(\p')}{\partial x_j}$ compared to the polynomial $\frac{\partial h}{\partial x_j}-\frac{\partial h(\p')}{\partial x_j}$ which vanishes at $\p'$.

For other cases, if
$$\rank(\J(\H,\frac{\partial h}{\partial x_j})(\tilde{\p}'))=s+1 \text{ and } \frac{\partial h(\p')}{\partial x_j}=0,$$
it is clear that $\frac{\partial h}{\partial x_j}$ vanishes at $\p'$. Thus it is exact.
If $$\rank(\J(\H,\frac{\partial h}{\partial x_j})(\tilde{\p}'))=s (\forall j),$$
according to our constructive method, we should do the linear combination $$f=\frac{\partial h}{\partial x_j}+\sum_{i=1}^{s}\alpha_i h_i \,\, (\text{ for some $j$ })$$
and compute its derivatives. Thus the perturbed term $\frac{\partial h(\p')}{\partial x_j}$ disappears. We will get an exact polynomial which vanishes at $\p'$ in the end. Notice that if $h_i$'s have perturbed terms, which are constants $h_i(\p')$. We know that if we compute the derivatives of $f$, these terms will disappear. Thus whether $h_i$'s have perturbed terms or not, the polynomials in the final deflated system derived by the linear combinations vanish at the exact zero $\p'$.


Now we consider the case that $h$ is regarded as $\theta$-regular at $\tilde{\p}'$ while it is singular at $\p'$. If
$$\rank(\J(\H,h)(\tilde{\p}'))=s,$$
we will do the linear combination of $h$ and $h_1,\ldots,h_s$ and compute its derivatives. It is obvious that this operation has no influence on our result. Usually the case
$$\rank(\J(\H,h)(\tilde{\p}'))=s+1$$
will not happen. It is related to the numerical computation of the rank of the Jacobian matrix of $(\H,h)$ at $\tilde{\p}'$.


As a summary of the foregoing analysis, we have:

\vspace{3mm}
Let $\F=\{f_1,\ldots,f_n\}\subset\C[\x]$ be a polynomial system. $\tilde{\p}\in\C^n$ is an approximate zero of $\F=\0$ and $\theta$ is a tolerance. According to our method, we acquire a final system $\widetilde{\F}'\subset\C[\x,\alphahei]$. During we compute the final system $\widetilde{\F}'$,
\begin{enumerate}
  \item if we judge a polynomial, which is singular at the exact zero $\p$, as being $\theta$-regular at $\tilde{\p}$, the final system $\widetilde{\F}'$ is accurate.
  \item if we judge a polynomial, which is regular at the exact zero $\p$, as being $\theta$-singular at $\tilde{\p}$, the final system $\widetilde{\F}'=\widetilde{\F}+\bm{\vartheta}$, is a perturbed system, where $\widetilde{\F}$ is an accurate system and $\bm{\vartheta}$ is the perturbed term, which satisfies $\max\limits_i|\vartheta_i|<\theta$.
\end{enumerate}

In actual computation, to make our method as accurate as possible, we give an adaptive adjustment step at the end of our algorithm. To be specific,
assume that the initial tolerance $\theta=\theta_1$. After the refining steps, denote the refined zero as $\bar{\p}$. We compute the Taylor expansions of all the related polynomials in computing the system $\widetilde{\F}'$ at $\bar{\p}$, including all the input polynomials. We denote the maximal absolute value of both the coefficients of the polynomials, which are judged as $\theta_1$-singular at $\tilde{\p}$ and the polynomials, which are judged as $\theta_1$-regular at $\tilde{\p}$, as $\theta_2$. It is also the term named ``Max err" in Tables 1 and 2 in the next section.

It is easy to imagine that $\theta_2\le\theta_1$ usually. If $\theta_2$ has a very higher precision than $\theta_1$, such as $\theta_1=10^{-2}$ and $\theta_2=10^{-13}$, we are sure that our conclusion is exact. If $\theta_2>\theta_1$ or $\theta_2$ still has a bad accuracy, such as $\theta_1=10^{-2}$ and $\theta_2=10^{-1}$ or $\theta_2=10^{-4}$,  we will take a smaller $\theta<\min\{\theta_1,\theta_2\}$ and repeat our method again.

After repeating our method several times, if $\theta_2$ is still bad, we will merely get a perturbed system.
\vspace{3mm}

Now, we give two examples to explain the above analysis.

\begin{exm}
  Given a polynomial system $\F=\{f_1=x+x^2+10000y^2,f_2=x^2+10000y^2\}$ with an approximate zero $$\tilde{\p}=(\tilde{p}_1,\tilde{p}_2)=(0.0006851,-0.0004368).$$ Consider the tolerances $\varepsilon=0.05$ and $\theta=0.5$. By the Taylor expansions of $f_i$ at $\tilde{\p}$, we know that $f_1,f_2$ are both $\theta$-regular at $\tilde{\p}$.

  Next, according to Algorithm 1, we compute $$\rank(\J(\F)(\tilde{\p}),\varepsilon)=2.$$ Thus, we can use Newton's method to refine $\tilde{\p}$ to a higher accuracy and get $$\tilde{\p}'=(0.0000000001,-0.0000008533).$$
  At this time, it's easy to check that $f_1$ is $\theta$-regular at $\tilde{\p}'$ and $f_2$ is $\theta$-singular at $\tilde{\p}'$. Therefore, for $f_2$, we have $$\frac{\partial f_2}{\partial x}=2x,\ \frac{\partial f_2}{\partial y}=20000y,$$ which are both $\theta$-regular at $\tilde{\p}'$. Furthermore, $$\rank(\J(f_1,\frac{\partial f_2}{\partial y}),\varepsilon)=2.$$ Thus, we get the final system $\widetilde{\F}'=\{f_1,\ 20000y\}$. After applying Newton's method, we get the refined zero $\bar{\p}=(\bar{p}_1,\bar{p}_2)=10^{-16}\cdot(0.53016,0)$.

  At last, we check if our chosen $\theta$ is proper. We compute the Taylor expansion of all the polynomials, which is judged as $\theta$-singular at $\tilde{\p}$, at the refined zero $\bar{\p}$ and get:
  $$f_2=2.810696256\cdot10^{-33}+1.060320\cdot10^{-16}\cdot(x-\bar{p}_1)+(x-\bar{p}_1)^2+20000\cdot (y-\bar{p}_1)^2.$$
  Thus, we have $$\mathrm{Max\ err}:=\max\{2.810696256\cdot10^{-33},1.060320\cdot10^{-16}\}=1.060320\cdot10^{-16}\ll\theta,$$
  which means that our final system $\widetilde{\F}'$ is more accurate than before.

\end{exm}

\begin{exm}\label{exm_perturb}
  Consider the system $\F=\{f_1=x+x^2+2xy+10000y^2,f_2=\frac{1}{20}x+x^2+2xy+10000y^2\}$ with an approximate zero $$\tilde{\p}=(\tilde{p}_1,\tilde{p}_2)=(0.000006851,-0.000004368).$$ Let the tolerances $\varepsilon=0.05$ and $\theta=0.5$. Similarly, by the Taylor expansions of $f_i$ at $\tilde{\p}$, we know that $f_1$ is $\theta$-regular at $\tilde{\p}$ and $f_2$ is $\theta$-singular at $\tilde{\p}$. Therefore, we have $$\frac{\partial f_2}{\partial x}=\frac{1}{20}+2x+2y,\ \frac{\partial f_2}{\partial y}=2x+20000y.$$
  Compute $$\rank(\J(f_1,\frac{\partial f_2}{\partial x}),\varepsilon)=2.$$ Thus, we get the final system $$\widetilde{\F}_1'=\{f_1,\ \frac{1}{20}+2x+2y\}.$$
  Obviously, $\widetilde{\F}_1'$ is a perturbed system and $\vartheta_2=\frac{1}{20}$ is the perturbed term, which satisfies $|{\vartheta}_2|<\theta$. It's easy to imagine that with $\widetilde{\F}_1'$, we could not get a good result. The main reason is that $\theta=0.5$ is too big, which leads to a wrong judgement on whether $f_2$ is $\theta$-regular at $\tilde{\p}$.

  If given another smaller tolerance $\theta'=0.05$, we will get a right judgement that $f_2$ is $\theta'$-regular at $\tilde{\p}$. Thus, we consider the linear combination of $f_1$ and $f_2$. Let $f=f_2+\alpha f_1$ and compute
  \begin{eqnarray*}
    g_1=\frac{\partial f}{\partial x} &=& \frac{1}{20}+2x+2y+\alpha(2x+2y+1), \\
    g_2=\frac{\partial f}{\partial y} &=& 2x+20000y+\alpha(2x+20000y),
  \end{eqnarray*}
  where $\alpha$ is a new variable and its initial value $\tilde{\alpha}=-0.050076986$. Compute $$\rank(\J(f_1,g_1,g_2),\varepsilon)=3.$$ Thus, we get the final system $\widetilde{\F}'=\{f_1,g_1,g_2\}$. Similarly, we consider applying Newton's method on the final system $\widetilde{\F}'$ and get the refined zero:
  $$\bar{\p}=(0.000000000000000,0.000000000000000,-0.050000000000000).$$
  Then, we check the coefficients of the terms with degree one of the Taylor expansion of $f$ at $\bar{\p}$ and get $$\mathrm{Max\ err}:=\{0,0,0\}=0\ll\theta'=0.05.$$
  Thus, we are sure that our final system $\widetilde{\F}'$ is accurate. Here, ``0" is not exact zero but means in Matlab machine accuracy.

\end{exm}

From the above two examples, we can see that once given an appropriate tolerance $\theta$, we can make sure that our final system is accurate. Otherwise, what we acquired is just a perturbed system, such as the system $\widetilde{\F}'_1$ in Example \ref{exm_perturb}.

Next, we continue analyzing the other tolerance $\varepsilon$, which is used to judge the numerical rank of a matrix. That is, we determine the numerical rank by comparing the absolute values of the singular values of the Jacobian matrix at approximate zero with the tolerance $\varepsilon$. Specifically, assume that we have computed an intermediate system $\H=\{h_1,\ldots,h_s\}\subset\C[\x']$. Denote $\x'=(\x,\alphahei)$. Assume that $\p$ is an isolated singular zero of the original system. The exact value of $\alphahei$ related to the coefficients of linear combinations is $\hat{\alphahei}$. Denote $\p'=(\p,\hat{\alphahei})\in\C^t$. Let $\tilde{\p}'\in\C^t$ be an approximate zero of $\H$ related to $\p'$ such that
$$\rank(\J(\H)(\tilde{\p}'),\varepsilon)=s.$$

Next, we consider one more polynomial $h_{s+1}\in\C[\x']$. Given the tolerance $\theta$, we can compute a polynomial $h$ from $h_{s+1}$, which is $\theta$-regular at $\tilde{\p}'$.
Denote $$\rank(\J(h_1,\ldots,h_s,h)(\p'))=r_1,\ \ \ \rank(\J(h_1,\ldots,h_s,h)(\tilde{\p}'),\varepsilon)=r_2.$$
For simplicity, we denote the deflated system as $\H'$, which comes from $\{\H,\, h_{s+1}\}$ after one step deflation, and its corresponding exact zero as $\q$, whose partial projection is $\p'$.

 According to the above analysis of $\theta$, for $h_{s+1}$, we have the following cases:
\begin{enumerate}
  \item if $\theta$ is chosen properly, that is, we judge $h_{s+1}$, which is regular or singular at $\p'$, as being $\theta$-regular or $\theta$-singular at $\tilde{\p}'$ respectively, we know that $h$ is regular at $\p'$. Thus, we have:
  \begin{enumerate}[(a).]
    \item  if $r_2=r_1$, we, of course, get an exact system $\H'$. That is, $\H'(\q)=\0$.
    \item  if $r_2<r_1$, according to our algorithm, we consider do the linear combination: $$g=h+\sum\limits_{j=1}^{s}\alpha_jh_j$$ and compute all the derivatives of $g$ with respect to all variables: $g_i=\frac{\partial g}{\partial x_i'},\ i=1,\ldots,t$. Correspondingly, $\H'=\{h_1,\ldots,h_s,g_1,\ldots,g_t\}$. Note that $\J(h)(\p')$ and $\J(h_1)(\p'),\ldots,\J(h_s)(\p')$ are actually linear independent, which means that the equations $(\alpha_1,\ldots,\alpha_s,1)\,\J(\H)(\p')=\0$ has no solution. Thus, although we can give the initial value $\tilde{\alpha}_j$ of $\alpha_j$ by solving a Least Squares problem, the linear independent will bring us some inexact polynomials $g_i$, which means $g_i(\q')\neq0$. Further, we may get a perturbed system $\H'$. That is, $\H'(\q)\neq\0$.
    \item  if $r_1<r_2$, we will add $h$ to the system $\H$ directly and get an exact system $\H'=\{h_1,\ldots,h_s,h\}$.
  \end{enumerate}
  \item if $\theta$ is chosen too big, that is, we judge $h_{s+1}$, which is regular at $\p'$, as being $\theta$-singular at $\tilde{\p}'$, we may get a perturbed polynomial $h$, which means that $h(\p')\neq 0$ and $h-h(\p')$ is regular at $\p'$. Thus, we have:
  \begin{enumerate}[(a).]
    \item if $r_2=r_1$, only the choice of $\theta$ affects our final result. Thus, we may get a perturbed system $\H'$ in this case.
    \item if $r_2<r_1$, with a similar discussion with the case of 1(b), we get a perturbed system $\H'$.
    \item if $r_1<r_2$, we add $h$ to $\{h_1,\ldots,h_s\}$ directly and get a perturbed system $\H'$.
  \end{enumerate}
  \item if $\theta$ is chosen too small, that is, we judge $h_{s+1}$, which is singular at $\p'$, as being $\theta$-regular at $\tilde{\p}'$, we know that $h=h_{s+1}$ is singular at $\p'$. Thus, we have:
  \begin{enumerate}
    \item if $r_2=r_1$, only the choice of $\theta$ affects our result. Thus, the system $\H'$ is exact in this case.
    \item if $r_2<r_1$, similar to the case of 1(b), we consider do the linear combination: $$g=h+\sum\limits_{j=1}^{s}\alpha_jh_j.$$ One difference from 1(b). is that $h$ is singular at $\p'$. Thus, all the $g_i=\frac{\partial g}{\partial x_i'}$ are exactly vanish at $\q$. So, the system $\H'$ is also exact in this case.
    \item if $r_1<r_2$, we add $h$ to $\{h_1,\ldots,h_s\}$ directly and get an exact system $\H'$.
  \end{enumerate}
\end{enumerate}

With the above analysis, we know that the choice of the tolerances $\theta$ and $\varepsilon$ has an influence on our final deflated system: an exact deflated system or a perturbed system. To judge which case a final deflated system belongs to, we use the following judgment method:

Denote the input system as $\F=\{f_1,\ldots,f_n\}$, the final deflated system $\widetilde{\F}'$. Noticing that we use Newton's method to refine the system $\widetilde{\F}'$, thus, we denote Newton's iteration sequence as $\{\tilde{\p}_l,\, l\ge 1\}$ and the final certified zero $\tilde{\p}'$.
\begin{itemize}
  \item First, we check if Newton's iteration sequence $\{\tilde{\p}_l,\, l\ge 1\}$ is quadratic convergence. If not, we claim that our deflated system is a perturbed system.
  \item If it is, we compute
  $$\Delta:=\max\{|f_i({\tilde{\p}'})|\,|f_i\in\F,\ i=1,\ldots,n\}. $$
  \item Next, we give a tolerance $\theta'$, which is usually a very small value, and compare the magnitude of $\theta'$ and $\Delta$. If $\Delta<\theta'$, we regard the final deflated system $\widetilde{\F}'$ as an exact system; otherwise, $\widetilde{\F}'$ is a perturbed system.
\end{itemize}

Of course, for the exact case, we are done. For the perturbed case, we hope to make our final deflated system as accurate as possible by adjusting the values of $\theta$ and $\varepsilon$. However, we still do not have a good idea on how to distinguish the effect of the two tolerances  $\theta$ and $\varepsilon$ on the final system. Fortunately, noting that the tolerance $\theta$ is used to accelerate our algorithm and is not necessary, therefore, according to the remark of Lemma \ref{lem_lvz}, we can use the deflation construction \ref{eq-G} to compute the final system. In this case, we just need to consider the tolerance $\varepsilon$, which is used to judge the numerical rank of the Jacobian matrix. That's to say, even if the first two steps of Algorithm \ref{alg-subalg} are removed, our deflation construction process can still work well.
 Considering the possible judgment, our final system can also be a perturbed system. Next, we give a possible modified method to overcome this case.

let $\F=\{f_1,\ldots,f_n\}$ be the input system, $\tilde{\p}\in\C^n$ be the initial approximate zero. Let $\varepsilon$ and $\theta'$ be the given tolerances.
\begin{itemize}
  \item First, assume $\rank(\J(\F)(\tilde{\p}),\varepsilon)=n$. We apply Newton's method on the system $\F$ and get the refined zero $\tilde{\p}'$. Next, we check if Newton's iteration sequence is quadratic convergence. If it is, we continue comparing the magnitude of $\theta'$ and $\Delta$. If $\Delta<\theta'$, we regard $\F$ as a system with an isolated simple zero; otherwise, we know $\F$ is a system with a multiple zero and $\rank(\J(\F)(\tilde{\p}),\varepsilon)<n$.
  \item Assume $\rank(\J(\F)(\tilde{\p}),\varepsilon)=n-1$. After using the deflation construction in Algorithm \ref{alg-subalg} once(from step 3 to step 10), we get a deflation system $\widetilde{\F}_1$ and an approximate zero $\tilde{\p}_1$. Then, we consider all the possibilities of $\rank(\J(\widetilde{\F}_1)(\tilde{\p}_1),\varepsilon)$. For every case, we go on our deflation construction in Algorithm \ref{alg-subalg} and use our mentioned judging method to check which case the final deflated system belongs to. As long as the final deflated system is judged to be an exact system, we will stop our deflation process; Otherwise, we know $\rank(\J(\F)(\tilde{\p}),\varepsilon)<n-1$.
  \item Assume $\rank(\J(\F)(\tilde{\p}),\varepsilon)=n-2$. We consider as the case of $n-1$.
\end{itemize}
About the above judgment process, we have two things to say:
\begin{itemize}
  \item[1.] The above judgement process must terminate in finite steps considering that our deflation construction terminates in finite steps.
  \item[2.] In the above judgement process, we traverse all the possibilities of the rank of the Jacobian matrix. Thus, there must be at least one case that we get an exact deflated system.
\end{itemize}

Now, we give an example to illustrate our idea.
\begin{exm}
  Continue with Example \ref{exm_perturb}. Here, we only use the tolerance $\varepsilon=0.05$ to judge the numerical rank. First, we compute $$\rank(\J(f_1,f_2)(\tilde{\p}),\varepsilon)=2.$$ Thus, we consider using Newton's method on the system $\F$ directly. Given an iterative error $10^{-8}$, we have the following Newton's iteration sequence:
  \begin{table}[H]
\small
\center
\begin{tabular}{cccc}
\hline
 $\p_i$ & x & y  \\
\hline
    $\tilde{\p}_1$   &   $\ 0.000006851$    &     $-0.000004368$  \\
    $\tilde{\p}_2$   &   $\ 0.0000000000000$    &     $-0.0000021841948$  \\
    $\tilde{\p}_3$   &   $-0.0000000000000$    &     $-0.0000010920974$  \\
    $\tilde{\p}_4$   &   $-0.0000000000000$    &     $-0.0000005460487$  \\
    $\tilde{\p}_5$   &   $-0.0000000000000$    &     $-0.0000002730243$  \\
    $\tilde{\p}_6$   &   $-0.0000000000000$    &     $-0.0000001365122$  \\
    $\tilde{\p}_7$   &   $-0.0000000000000$    &     $-0.0000000682561$  \\
    $\tilde{\p}_8$   &   $\ 0.0000000000000$    &     $-0.0000000341280$  \\
    $\tilde{\p}_9$   &   $-0.0000000000000$    &     $-0.0000000170640$  \\
    $\tilde{\p}_{10}$   &   $-0.0000000000000$    &     $-0.0000000853201$  \\
\hline
\end{tabular}
  \end{table}
 We can check easily that Newton's iteration sequence $\{\tilde{\p}_j,\ j=1,\ldots,10\}$ is linear convergence. According to our judging criteria, we know that $$\rank(\J(f_1,f_2)(\tilde{\p}),\varepsilon)=1.$$

  Next, according to our construction process in Algorithm \ref{alg-subalg}, we let $$g=f_2+\alpha f_1$$ and compute $$g_1=\J_1(g)=\alpha(2x+2y+1)+(1/20+2x+2y),\ g_2=\alpha(2x+20000y)+(2x+20000y).$$
  We have $\tilde{\alpha}=-0.091484814324$.

  Next, let $\widetilde{\F}_1=\{f_1,g_1,g_2\}$ and $\tilde{\p}_1=(\tilde{\p},\tilde{\alpha})$. By our given revised method above, we continue considering all the possibilities of $\rank(\J(\widetilde{\F}_1)(\tilde{\p}_1),\varepsilon)$. For example, we consider  the  case of $$\rank(\J(\widetilde{\F}_1)(\tilde{\p}_1),\varepsilon)=3.$$
  Similarly, given the iterative error $10^{-8}$, by using Newton's method on the system $\widetilde{\F}_1$, we get the following iteration sequence:

\begin{table}[H]
\small
\center
\begin{tabular}{cccc}
\hline
 $\p_i$ & x & y & $\alpha$ \\
\hline
 $\tilde{\p}_1$    &   0.000006851 &  $-0.000004368$  &  $-0.091484814324 $  \\
 $\tilde{\p}_2$    &   0.0000002081968 &  \ 0.0000001993959 & $-0.0500009466172 $   \\
 $\tilde{\p}_3$    &   0.0000000003977 & $-0.0000000000002$ & $-0.0500000007560$ \\
  $\tilde{\p}_4$    &   0.0000000000000 &  \ 0.0000000000000 & $-0.0500000000000$ \\
\hline
\end{tabular}
  \end{table}
It is easy to check that the iteration sequence $\{\tilde{\p}_j,j=1,2,3,4\}$ is quadratic convergence. Furthermore, given a tolerance $\theta'=10^{-12}$, we can compute
$$\Delta:=\max\{|f_1({\tilde{\p}_4})|,\ |f_2({\tilde{\p}_4})|\}=0 $$ and verify that $\Delta<\theta'$. Thus, we regard the final deflated system $\widetilde{\F}'=\widetilde{\F}_1$ as an exact system. At the same time, we stop our deflation process.

\end{exm}

%
%
%
%
%
%

Until now, we have finish all the discussions about the tolerances $\theta$ and $\varepsilon$.
 Once given a polynomial system $\F\subset\C[\x]$ with an isolated singular zero, we use Algorithm \ref{alg-subalg} to compute a new system $\widetilde{\F}'(\x,\alphahei)$, which has a simple zero. What's more, according to the analysis of the tolerances $\theta$ and $\varepsilon$, our final system $\widetilde{\F}'$ is an accurate system usually. For the perturbed case, we also give one ergodic way to adjust our final result as accurate as possible. Thus, we can use the final system $\widetilde{\F}'$ to certify the isolated zeros of the input system.

In the following, by using the algorithm  {\bf verifynlss} in INTLAB\cite{Rump3}, we give an example to explain how we certify the isolated singular zero of the input system.
\begin{exm}\label{exm-3}

Continue with Example \ref{exm-4}. Applying the algorithm {\bf verifynlss}, we get the system: 
\begin{eqnarray*}
\widetilde{\F}(\x,\alphahei)=\left\{\aligned
   \tilde{f}_1&=-\frac{9}{4}+ \frac{3}{2}x_1+2x_2+3x_3+4x_4-\frac{1}{4}x_1^2,  \\
   \tilde{f}_2&=x_1-2x_2-2x_3-4x_4+2x_1x_2+3x_1x_3+4x_1x_4,  \\
   \tilde{f}_3&=-8+4x_1+4x_4-2x_1x_4, \\
   \tilde{f}_4&=3+\frac{3}{2}\alpha_1+\alpha_2+4\alpha_3-\frac{1}{2}\alpha_1x_1+2\alpha_2x_2+3\alpha_2x_3+4\alpha_2x_4-2\alpha_3x_4,\\
   \tilde{f}_5&=2+2\alpha_1-2\alpha_2+2\alpha_2x_1,  \\
   \tilde{f}_6&=4+3\alpha_1-2\alpha_2+3\alpha_2x_1,  \\
   \tilde{f}_7&=4+4\alpha_1-4\alpha_2+4\alpha_3+4\alpha_2x_1-2\alpha_3x_1,
  \endaligned\right.
\end{eqnarray*}
%
%
%
%
and two verified inclusions
{\small
$$
\boldsymbol{\mathbf{X}}=
\left[
\begin{array}{cccccccccccccccccc}
  [\ \ 0.99999999999999,\ \ \ 1.00000000000001] \\

  [-2.00000000000001,\ \  -1.99999999999998] \\

  [-1.00000000000001,\ \ -0.99999999999999] \\

  [\ \ 1.99999999999999,\ \ \ 2.00000000000001]
%
\end{array}
\right]
$$
and
$$
\boldsymbol{\mathcal{A}}=
\left[
\begin{array}{cccccccccccccccccc}
  [-1.00000000000001,\ -0.99999999999999] \\

  [-1.00000000000001,\ -0.99999999999999] \\

  [-0.00000000000001,\ -0.00000000000001]

\end{array}
\right]
.$$
}

For the deflated system $\widetilde{\F}(\x,\alphahei)$, we affirm that there is a unique isolated simple zero $(\hat{\x},\hat{\alphahei})\in(\X,\boldsymbol{\mathcal{A}})$, such that $\widetilde{\F}(\hat{\x},\hat{\alphahei})=\0$.
What's more, the projection $\hat{\x}$ of $(\hat{\x},\hat{\alphahei})$ corresponds to the isolated singular zero of the input system $\F$. That's to say, we certified the isolated singular zeros of the original system.
\end{exm}

\section{Experiments and results}
We implement our method in Matlab of Algorithm \ref{alg-subalg}. The code and some examples can be found in \url{http://www.mmrc.iss.ac.cn/~jcheng/VDSS}.
In this section, we show the results of the experiment and the comparison of our method with some other methods. We do the experiments in Matlab R2012b with INTLAB-V5.5 on a computer with Windows 7, Intel $i7$ processor and $8GB$ memory.

in \cite{zhi3}, by modifying the method proposed by Yamamoto \cite{Yama}, they give a deflation method to compute a regular and square augmented system. which can be used to prove the existence of an isolated singular solution of a slightly perturbed system. Moreover, by applying INTLAB function {\bf verifynlss}\cite{Rump3}, they also give an algorithm $\mathbf{viss}$ to compute verified error bounds. However, noticing that their method is essentially a deflation method. Thus, we also implement our algorithm based on INTLAB function {\bf verifynlss}.

In Table 1, we compare our algorithm $\mathbf{VDSS}$ with the algorithm $\mathbf{viss}$. These examples are relatively simple and small scale, which can be found in \cite{dayton1,zhi3}. We also list them in \url{http://www.mmrc.iss.ac.cn/~jcheng/VDSS/fun.m}. We denote $var$ the number of polynomials, $mul$ the multiplicity and Verified acc the final verified accuracy, which is measured by the breath of the verified inclusion $\X$. And Max err is $\delta_2$ as mentioned in Section 4.2.
We use a same initial accuracy $10^{-4}$ for all the examples. ``true" means we get two same endpoints of the verified inclusion. When the term for Max err is ``0", it does not mean Max err is exactly zero and only shows in Matlab machine precision.

{\small
\begin{table}[H]
\caption{Comparison of {\bf VDSS} and {\bf viss} for simple systems}
\centering
\begin{tabular}{|c|c|c|c|c|c|c|c|c|c|c|}
\hline
\multicolumn{1}{|c|}{\multirow {2}{*}{System}}&
\multicolumn{1}{c|}{\multirow {2}{*}{$var$}}&
\multicolumn{1}{c|}{\multirow {2}{*}{$mul$}}&
\multicolumn{2}{c|}{Verified acc}&
\multicolumn{1}{c|}{\multirow {2}{*}{Max err}}&
\multicolumn{2}{c|}{times}&
\multicolumn{2}{c|}{Final size} \\
\cline{4-5}
\cline{7-10}

\multicolumn{1}{|c|}{}&\multicolumn{1}{c|}{}&\multicolumn{1}{c|}{}&
\multicolumn{1}{c|}{VDSS} & viss &\multicolumn{1}{c|}{}& \multicolumn{1}{c|}{VDSS} & viss & \multicolumn{1}{c|}{VDSS} & viss \\
\hline
DZ1 & 4 & 131  & true & e-322 & 0 & 0.3066 & 0.3337 & 4 & 16  \\ \hline
DZ2 & 3 & 16  & e-14 & e-14 & 0 & 0.2989 & 0.7343 & 3 & 24 \\ \hline
DZ3 & 2 & 4   & e-14 & e-15 & e-14 & 0.8780 & 1.0093 & 3 & 10 \\ \hline
cbms1 & 3 & 11  & true & e-322 & 0  & 0.1851 & 0.1107 & 3 & 6 \\ \hline
cbms2 & 3 & 8   & true & e-322 & 0 & 0.2546 & 0.1271 & 3 & 6  \\ \hline
mth191 & 3 & 4  & e-14 & e-14 & e-32 & 0.3118 & 0.1221 & 4 & 6  \\ \hline
KSS & 10 & 638   & e-14 & e-14 & 0 & 8.2295 & 0.3036 & 19 & 20 \\ \hline
RuGr09 & 2 & 4  & e-323 & e-14 & 0 & 0.1567 & 0.4955 & 2 & 8 \\ \hline
LZ & 100 & 3  & e-320 & e-14 & 0 & 2.0197 & 13.3068 & 100 & 300 \\ \hline
Ojika1 & 2 & 3 & e-14  & e-14 & 0 & 0.7636 & 0.3447 & 5 & 6  \\ \hline
Ojika2 & 3 & 2  & e-14 & e-14 & e-16 & 0.3936 & 0.2942 & 5 & 6 \\ \hline
Ojika3 & 3 & 2  & e-14 & e-14 & 0 & 0.3967 & 0.3427 & 4 & 6 \\ \hline
Ojika4 & 3 & 3   & e-14 & e-14 & 0 & 0.1851 & 1.0621 & 3 & 9 \\ \hline
Decker2 & 3 & 4   & e-323 & e-14 & 0 & 0.1752 & 0.4650 & 3 & 8 \\ \hline
Caprasse & 4 & 4   & e-14 & e-14 & e-31 & 2.0180 & 0.5126 & 6 & 8 \\ \hline
Cyclic9 & 9 & 4   & e-14 & e-14 & e-15 & 5.9266 & 3.6878 & 12 & 18 \\ \hline
\end{tabular}
\end{table}}

From Table 1, we can see that our algorithm is effective. On one hand, the verified accuracy of our method is never worse than $\mathbf{viss}$ for all these examples.   
On the other hand, thanks to our acceleration strategies, our practical size and computing time are smaller than those of $\mathbf{viss}$ in most cases.

We also compare our method with {\bf viss} for large-scale polynomial systems.
 All the examples in Table 2 can be found in \url{http://www.mmrc.iss.ac.cn/~jcheng/VDSS/example.m}.
 The example LZ2000 can be found in \cite{zhi1}. The example nonpoly3 is a non-polynomial nonlinear system. We construct the other examples as below: First, we produce some polynomials randomly to form a zero-dimensional system $\{f_1,\ldots,f_n\}$, which has a simple zero $\p$ and $\deg(f_i)\ge 2$ usually.
 The final systems have the form:
 $\F=\{f_i^{d_i}+g_i,1\le i\le n, g_i\in\{f_1^{d'_1},\ldots,f_n^{d'_n},0\},d_i\ge 1, d'_i\ge 1,1\le i\le n\}$. The new systems are always dense polynomial systems. The examples named  simple1, reduce3, big1, big2, big3, large3, large6, large8 are of the form that $g_i=0 (1\le i \le n)$; The examples named addvar3, unre3, unre5, rankone2, rankone3 are of the form that $g_i$ are not all zeros. The ranks of the Jacobian matrices of the examples rankone2, rankone3 at the zeros both are one.
 In Table 2, ``-" means there is no results with the code. 
{\small
\begin{table}[H]
\caption{Comparison of {\bf VDSS} and {\bf viss} for large systems}
\centering
\begin{tabular}{|c|c|c|c|c|c|c|c|c|c|c|c|c|}
\hline
\multicolumn{1}{|c|}{\multirow {2}{*}{System}}&
\multicolumn{1}{c|}{\multirow {2}{*}{$var$}}&
\multicolumn{1}{c|}{\multirow {2}{*}{$mul$}}&

\multicolumn{2}{c|}{Verified acc}&
\multicolumn{1}{c|}{\multirow {2}{*}{Max err}}&
\multicolumn{2}{c|}{times}&
 \multicolumn{2}{c|}{Final size} \\
\cline{4-5}
\cline{7-10}

\multicolumn{1}{|c|}{}&\multicolumn{1}{c|}{}&\multicolumn{1}{c|}{}&
\multicolumn{1}{c|}{VDSS} & viss &\multicolumn{1}{c|}{}& \multicolumn{1}{c|}{VDSS} & viss & \multicolumn{1}{c|}{VDSS} & viss \\
\hline
LZ2000 & 2000 & 3   & e-319 & -- & 0 & 448.07 & -- &2000 & --\\ \hline
simple1 & 5 & 9   & e-14 & e-14  & 0 & 0.29 & 8.20 &5 & 45\\ \hline
addvar2 & 4 & 12   & e-14 & e-13  & e-14 & 11.10 & 250.67 &6 & 32\\ \hline
reduce3 & 4 & 24  & e-14 & e-14  & e-11 & 12.21 & 317.50  &7 & 12\\ \hline
unre3 & 4 & 36   & e-15 & e-14  & e-13 & 4.08 & 360.32 &4 & 32\\ \hline
unre5 & 8 & 576   & e-14 & e-14  & e-13 & 24.26 & 229.83 &8 & 64\\ \hline
big1 & 20 & 512  & e-14 & e-15  & e-12 & 29.92 & 1724.09 &20 & 160\\ \hline
big2 & 20 & 8192  & e-14 & e-14  & e-12 & 40.90 & 1751.61  &20 & 160\\ \hline
big3 & 30 & 196608   & e-15 & e-14  & e-15 & 155.18 & 425.51  &30 & 240\\ \hline
rankone2 & 6 & 32   & e-15 & e-15  & e-15 & 6.8693 & 1.5199 &11 & 12\\ \hline
rankone3 & 6 & 96   & e-15 & e-14  & e-14 & 12.44 & 136.54 &11 & 48\\ \hline
breadth2 & 5 & $2^5$   & e-322 & --  & 0 & 0.20 & -- &5 & --\\ \hline
large3 & 100 & $3^{100}$   & e-323 & e-319  & 0 & 187.88 & 647.86 &100 & 400\\ \hline
large6 & 500 & $4^{100}$   & e-321 & e-34  & 0 & 905.00 & 3262.78 &500 & 2000\\ \hline
large8 & 500 & $4^{300}$   & e-321 & --  & 0 & 1745.85 & -- &500 & --\\ \hline
nonpoly3 & 3 & 64  & e-322 & e-14  & 0 & 0.19 & 6.62 &3 & 36\\ \hline
\end{tabular}
\end{table}}

From Tables 2, we can see that for the examples with more variables and high multiplicity, our method has a better result regardless of the verified accuracy, computing time or the final scale.

We also test the example in \cite{mourrain2} with the form: $\{x_1^3-x_1^2-x_2^2,x_2^3+x_2^2-x_3,\ldots,x_{n-1}^3+x_{n-1}^2-x_n,x_n^2\}$. The example named breath2 in Table 2 has this form for $n=5$. The method in \cite{mourrain1} can compute this example for $n=6$ and it takes 659.59 seconds with the final size for 321 variables and 819 polynomials. We test the cases for $n=6, n=1000$ and $n=2000$ with our code, it takes 0.228965 seconds, 165.274439 seconds and 1036.773847 seconds respectively without introducing new variables.

For our method, although we introduce new variables, the size of our final deflated system is small in experiments. And further,
 we also compare our method with the other four deflation methods \cite{mourrain1} on the following four small systems.
\begin{enumerate}[1:]
  \item $\{x_1^4-x_2x_3x_4,\ x_2^4-x_1x_3x_4,\ x_3^4-x_1x_2x_4,\ x_4^4-x_1x_2x_3\}$ at $(0,0,0,0)$ with $\mu=131$;
  \item $\{x^4,\ x^2y+y^4,z+z^2-7x^3-8x^2\}$ at $(0,0,-1)$ with $\mu=16$;
  \item $\{14x+33y-3\sqrt5x^2-12\sqrt5xy-12\sqrt5y^2-6\sqrt5+x^3+6x^2y+12xy^2+8y^3+\sqrt7,41x-18y-\sqrt5+8x^3-12x^2y+6xy^2-y^3+12\sqrt7xy-12\sqrt7x^2-3\sqrt7y^2-6\sqrt7\}$ at $\p\approx(1.5055,0.36528)$ with $\mu=5$;
  \item $\{2x_1+2x_1^2+2x_2+2x_2^2+x_3^2-1,(x_1+x_2-x_3-1)^3-x_1^3,(2x_1^3+5x_2^2+10x_3+5x_3^2+5)^3-1000x_1^5\}$ at $(0,0,-1)$ with $\mu=18$.
\end{enumerate}

 The result (see also in \cite{mourrain1}) is below, where method A is in \cite{dayton1,lvz}, method B is in \cite{hauensten}, method C is in \cite{giusti3}, method D is in \cite{mourrain1}, method E is our method $\mathbf{VDSS}$. In Table 3, we denote $Poly$ the number of the polynomials of the final deflation system and $Var$ the number of the variables in the final deflation system. Noting that our final system does not always contain all the polynomials of the input system, therefore, we will contain the number of the different polynomials in the input system, which is not contained in the final system, into $Poly$.
\begin{table}[H]
\small
\center
\caption{Comparison of VDSS and other methods for four examples}
\begin{tabular}{|c|c|c|c|c|c|c|c|c|c|c|c|c|c|c|}
\hline
\multicolumn{1}{|c|}{\multirow {2}{*}{ }}&
\multicolumn{2}{c|}{Method A}&
\multicolumn{2}{c|}{Method B}&
\multicolumn{2}{c|}{Method C}&
\multicolumn{2}{c|}{Method D}&
\multicolumn{2}{c|}{Method E} \\
\cline{2-11}

\multicolumn{1}{|c|}{}&\multicolumn{1}{c|}{$Poly$}&{$Var$}&\multicolumn{1}{c|}{$Poly$}&{$Var$}&\multicolumn{1}{c|}{$Poly$}&{$Var$}&\multicolumn{1}{c|}{$Poly$}&{$Var$}&\multicolumn{1}{c|}{$Poly$}&{$Var$}\\
\hline
1 & 16 & 4  & 22 & 4 & 22 & 4 & 16 & 4 & 8 & 4  \\ \hline
2 & 24 & 11  & 11 & 3 & 12 & 3 & 12 & 3 & 5 & 3 \\ \hline
3 & 32 & 17  & 6 & 2 & 6 & 2 & 6 & 2 & 4 & 3 \\ \hline
4 & 96 & 41  & 54 & 3 & 54 & 3 & 22 & 3 & 5 & 3 \\ \hline

\end{tabular}
\end{table}
In Table 3, for system 1, 2 and 4, our method matches the best of the other four methods and simultaneously has a smallest deflated system in the five methods. For system 3, although our final system has one more variable than method D, we have less polynomials.

\section{Conclusion}
 In this paper, we develop a new deflation method for refining or verifying the isolated singular zeros of polynomial systems. Given a polynomial system $\F\subset\C[\x]$ with an isolated singular zero $\p$, by computing the derivatives of the input polynomials directly or the linear combinations of the related polynomials, we prove constructively that there exists a final deflated system $\widetilde{\F}'(\x,\alphahei)$, which has an isolated simple zero $(\p,\hat{\alphahei})$, whose partial projection corresponds to the isolated singular zero $\p$ of the input system $\F$. New variables $\alphahei$ are introduced to represent the coefficients of the linear combinations of the related polynomials to ensure the accuracy of the numerical implementation.

 Compared to the previous deflation methods, on one hand, our method also has an output size depending on the depth or the multiplicity of $\p$ in theory. On the other hand, thanks to the acceleration strategies we proposed in the paper, the size of the final system in our actual computations is much less than that we give in theory. The results of the experiments we conduct give a very persuasive argument for this.

 In order to essentially have a deeper understanding of our approach, we also give some further analysis of the tolerances $\theta$ and $\varepsilon$ we use. The results of the analysis tells us that our final system is a perturbed system with a bounded perturbation in the worst case. To make our final system as accurate as possible, we also analyse the case that the tolerance $\theta$ is not introduced.

\section*{\bf Acknowledgement} The work is partially supported by NSFC Grants 11471327.

\bibliography{bibfile}
\bibliographystyle{plain}

\end{document}